\pdfoutput=1
\documentclass[cleveref]{lipics-v2021}

\usepackage{graphicx}
\usepackage{amsmath}
\usepackage{amssymb}
\usepackage{tikz}
\usepackage{complexity}
\usepackage{comment}
\usepackage[color=blue!20]{todonotes}
\usepackage{xfrac}
\usepackage{numprint} 
\usepackage{tabularx} 
\usepackage[color=blue!20]{todonotes}
\usetikzlibrary{positioning}
 \usetikzlibrary { decorations.pathmorphing, decorations.pathreplacing, decorations.shapes, } 
  \usetikzlibrary {fit} 
  \usetikzlibrary{arrows.meta}

\usepackage{apxproof} 

\renewcommand{\mainbodyrepeatedtheorem}{\textup{($\star$)}}

\newtheoremrep{theorem}{Theorem}
\newtheoremrep{claim}[theorem]{Claim}
\newtheoremrep{lemma}[theorem]{Lemma}
\newtheoremrep{corollary}[theorem]{Corollary}
\newtheoremrep{observation}[theorem]{Observation}

\newcommand{\I}{\mathcal{I}}

\renewcommand{\leq}{\leqslant}
\renewcommand{\geq}{\geqslant}

\DeclareMathOperator*{\argmax}{arg\,max}
\DeclareMathOperator*{\argmin}{arg\,min}

\newcommand{\e}[2]{$(#1, #2)$}

\title{Generalizing Roberts' characterization of unit interval graphs}

\author{Virginia Ardévol Martínez}{Université Paris-Dauphine, PSL University, CNRS, LAMSADE, 75016 Paris, France}{virginia.ardevol-martinez@dauphine.psl.eu}{https://orcid.org/0000-0002-3703-2335}{}
\author{Romeo Rizzi}{Department of Computer Science, University of Verona, Italy}{romeo.rizzi@univr.it}{}{}
\author{Abdallah Saffidine}{University of New South Wales, Sydney, Australia}{abdallah.saffidine@gmail.com}{}{}
\author{Florian Sikora}{Université Paris-Dauphine, PSL University, CNRS, LAMSADE, 75016 Paris, France}{florian.sikora@dauphine.psl.eu}{https://orcid.org/0000-0003-2670-6258}{}
\author{Stéphane Vialette}{LIGM, CNRS, Univ Gustave Eiffel, F77454 Marne-la-Vallée, France}{stephane.vialette@univ-eiffel.fr}{https://orcid.org/0000-0003-2308-6970}{}

\authorrunning{V. Ardévol et al.}


\hideLIPIcs 

\nolinenumbers

\begin{document}



\acknowledgements{Part of this work was conducted when RR was an invited professor at Université Paris-Dauphine. This work was partially supported by the  ANR project ANR-21-CE48-0022 (``S-EX-AP-PE-AL'').}

\maketitle              
\begin{abstract}

For any natural number $d$, a graph $G$ is a (disjoint) $d$-interval graph if it is the intersection graph of (disjoint) $d$-intervals, the union of $d$ (disjoint) intervals on the real line. Two important subclasses of $d$-interval graphs are unit and balanced $d$-interval graphs (where every interval has unit length or all the intervals associated to a same vertex have the same length, respectively).
A celebrated result by Roberts gives a simple characterization of unit interval graphs being exactly claw-free interval graphs. Here, we study the generalization of this characterization for $d$-interval graphs.
In particular, we prove that for any $d \geq 2$, if $G$ is a $K_{1,2d+1}$-free interval graph, then $G$ is a unit $d$-interval graph. However, somehow surprisingly, under the same assumptions, $G$ is not always a \emph{disjoint} unit $d$-interval graph.
This implies that the class of disjoint unit $d$-interval graphs is strictly included in the class of unit $d$-interval graphs.
Finally, we study the relationships between the classes obtained under disjoint and non-disjoint $d$-intervals in the balanced case and show that the classes of disjoint balanced 2-intervals and balanced 2-intervals coincide, but this is no longer true for $d>2$.


\end{abstract}

\section{Introduction}

 Interval graphs are the intersection graphs of intervals on the real line: every vertex represents an interval and there is an edge between two vertices if and only if their corresponding intervals intersect. The class of interval graphs is one of the most important classes of intersection graphs, mostly due to their numerous applications in scheduling or allocation problems and in bioinformatics, see for examples these monographs~\cite{Fishburn1985,mckee1999topics,roberts1978graph}. 

Already in the late 70s, situations arising naturally in scheduling and allocation motivated the generalization of interval graphs to \emph{multiple interval graphs}, where every vertex is associated to the union of $d$ intervals on the real line (called a $d$-interval), for some natural number $d$, instead of to a single interval. This allowed a more robust modeling of problems such as multi-task scheduling or allocation of multiple associated linear resources~\cite{DBLP:journals/siammax/GriggsW80,mcguigan1977presentation,DBLP:journals/jgt/TrotterH79}, and led to several interesting problems~\cite{DBLP:journals/tcs/FellowsHRV09,DBLP:journals/algorithmica/Francis0O15,DBLP:journals/tcs/Jiang10,DBLP:journals/algorithmica/Jiang13,DBLP:journals/tcs/JiangZ12}. 
The applications of 2-interval graphs to bioinformatics also increased the interest on this class of graphs~\cite{joseph1992determining,vialette2004computational}.

These concrete applications of multiple interval graphs, specifically 2-interval graphs, suggested a focus on different restrictions, such as unit 2-interval graphs \cite{bar2006scheduling}, or balanced 2-interval graphs~\cite{DBLP:journals/tcs/CrochemoreHLRV08}.
For both interval and multiple interval graphs, we say that they are \emph{unit} if all the intervals in the representation, i.e. the set of intervals associated to the graph, have unit length. 
For multiple interval graphs, we also define the subclass of \emph{balanced} $d$-interval graphs, where all intervals forming the same $d$-interval have equal length, but intervals of different $d$-intervals can have different lengths. Finally, for both interval and multiple interval graphs, we say that they are \emph{proper} if there exists an interval representation where no interval properly contains another one.
The class of unit 2-interval graphs is known to be properly contained in the class of balanced 2-interval graphs~\cite{DBLP:conf/wg/GambetteV07}. 

Let us remark that in the literature, $d$-intervals have been defined both as the union of $d$ \emph{disjoint} intervals \cite{bar2006scheduling,butman2010optimization,DBLP:journals/dam/WestS84}, as the union of $d$ \emph{not necessarily disjoint} intervals \cite{DBLP:journals/jgt/TrotterH79}, and simply as the union of $d$ intervals, without specifying whether they are disjoint or not \cite{erdos1985note,scheinerman1983interval}. 
This ambiguity is not relevant in the general case, since both definitions lead to the same class of graphs. 
However, in this paper we focus on subclasses of multiple interval graphs, namely unit and balanced, for which this equivalence is not known to be true. 
Therefore, we will distinguish between the two possible definitions of $d$-intervals.
The first definition is denoted as \emph{disjoint $d$-intervals} while the second is simply denoted as \emph{$d$-intervals} (further details are discussed in Section~\ref{SectionDefinitions}). 
 
From an algorithmic perspective, another reason why interval graphs have been widely studied is because many problems that are \NP-hard become solvable in polynomial time when restricted to this class of graphs. This is not the case for $d$-interval graphs~\cite{bar2006scheduling,butman2010optimization,DBLP:journals/algorithmica/Francis0O15}. The problem of recognizing $d$-interval graphs is no exception: it is $\NP$-complete for every natural number $d\geq 2$ \cite{DBLP:journals/dam/WestS84}, even for unit 2-interval graphs~\cite{acceptedvirginia} and balanced 2-interval graphs~\cite{DBLP:conf/wg/GambetteV07long}. In sharp contrast, the recognition of interval graphs (both in the unit and unrestricted case) can be done in polynomial time~\cite{DBLP:journals/jcss/BoothL76,DBLP:journals/tcs/HabibMPV00}, and there exist multiple characterizations of them, including a characterization in terms of forbidden induced subgraphs~\cite{lekkeikerker1962representation,zbMATH03307330}. 
In particular, in 1969, Roberts proved that the class of proper interval graphs and the class of unit interval graphs coincide \cite{zbMATH03307330}, and showed that unit interval graphs are exactly $K_{1,3}$-free interval graphs (i.e., interval graphs that do not contain the star with three leaves as an induced subgraph). To do so, he used the Scott-Suppes characterization of semiorders (see \cite{DBLP:journals/dm/BogartW99,DBLP:journals/dm/Gardi07} for short constructive proofs of this result). This is a remarkable result as it gives a simple characterization of unit interval graphs. It also implies that if $G=(V,E)$ is a unit interval graph, then there exists a semiorder $S(V,P)$ on the vertices of $V$ such that $(u,w)\in P$ if and only if $(u,w) \notin E$, which justifies the original name of ``indifference graphs'' for unit interval graphs (as they can represent indifference relations by joining two elements by an edge if neither is preferred over the other one).

It is straight-forward to check that being $K_{1,3}$-free is a necessary condition for being a unit interval graph, as an interval of unit length cannot intersect three pairwise disjoint intervals of length one. The reader can observe that this necessary condition extends naturally to multiple interval graphs: a unit 2-interval graph cannot contain a $K_{1,5}$ as an induced subgraph; and more generally, a unit $d$-interval graph cannot contain a $K_{1,2d+1}$ as an induced subgraph. Thus the following natural question arises: can we generalize Roberts characterization of unit interval graphs to multiple interval graphs? 
Perhaps the most straight-forward generalization would be to characterize unit $d$-interval graphs as $K_{1,2d+1}$-free $d$-interval graphs, but this has already been proven false in \cite{alexandre}: there exists a graph which is 2-interval and $K_{1,5}$-free, but not unit 2-interval. But not all hope of generalizing Roberts characterization must be lost yet! What if we add some additional constraints? 

Already in 2016, Durán et al. decided to focus on $d$-interval graphs which are also \emph{interval}~\cite{Duran2016}. In a presentation at VII LAWCG, they claimed that if $G$ is an interval graph, then $G$ is a disjoint unit $d$-interval graph if and only if it is $K_{1,2d+1}$-free \footnote{Note that they refer to disjoint unit $d$-intervals simply as unit $d$-intervals, but they are explicitly defined beforehand as the union of $d$ disjoint intervals.}. 
In this paper, we show that the aforementioned statement is actually false, and that, perhaps surprisingly, Roberts characterization can only be generalized depending on the chosen definition of $d$-interval graphs! (See \cref{fig:recap classes} for a summary of the main results).

We also study the subclasses obtained under the two definitions of $d$-intervals in the balanced case, expanding the knowledge of the relationships between the different subclasses of 2-interval graphs. 

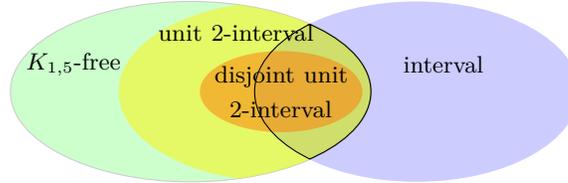
\begin{figure}[!ht]
\centering
\begin{tikzpicture}[scale=1.2]
    \draw[fill=green, opacity=0.2] (0,0) ellipse (2 and 1);
    
    \fill[fill=blue, opacity=0.2] (2.5,0) ellipse (1.8 and 1);
       
    \begin{scope}
       \fill[fill=red, opacity=0.5] (1,0) ellipse (0.9 and 0.45);
    \end{scope}
    
    \begin{scope}
        \clip (1,0) ellipse (1.8 and 1);

        \fill[fill=yellow, opacity=0.5] (0,0) ellipse (2 and 1);
    \end{scope}

    \begin{scope}
  \clip (0,0) ellipse (2 and 1);
  \clip (2.5,0) ellipse (1.8 and 1);
  \draw[thick, black] (2.5,0) ellipse (1.8 and 1);
\draw[thick, black] (0,0) ellipse (2 and 1);
\end{scope}

    \node at (-1.3,0.3) {\small $K_{1,5}$-free};
\node at (2.8,0.3) {\small interval};
\node[text width=2.5cm,align=center] at (1,0) {\small disjoint unit 2-interval};
\node[text width=2.5cm,align=center] at (0.5,0.65) { \small unit 2-interval};

\end{tikzpicture}

\caption{$K_{1,5}$-free interval graphs are not contained in the class of disjoint unit 2-interval graphs.
The class of unit 2-interval graphs is a superclass of disjoint unit 2-interval graphs, and spans the whole intersection of $K_{1,5}$-free and interval graphs.
\label{fig:recap classes}
}
\end{figure}
The structure of the paper is as follows: Section~\ref{SectionDefinitions} briefly introduces the necessary definitions and 
discusses the definition of $d$-interval. 
In Section~\ref{SectionAlgorithm}, we prove that if $G$ is an interval graph, then it is unit $d$-interval if and only if it is $K_{1,2d+1}$-free. We then show that this result cannot be generalized for \emph{disjoint} multiple intervals in Section~\ref{SectionGeneral}, which implies that the class of disjoint unit $d$-interval graphs is actually properly contained in the class of unit $d$-interval graphs.
Finally, we study the balanced case in Section~\ref{SectionBalanced}, and show that the definition of $d$-interval also matters, as the classes of disjoint balanced $2$-intervals and balanced $2$-intervals coincide, but this is no longer true for $d>2$. We conclude with some open questions in Section~\ref{SectionConclusion}.
Due to space constraints, some proofs, marked with a ($\star$), are deferred to the appendix.

\section{Preliminaries}\label{SectionDefinitions}

In the following, $G=(V,E)$ will denote a simple undirected graph on the set of vertices $V$ and with edges $E$, and an interval will be a set of real numbers of the form $[a,b] := \{x\in \mathbb{R} \mid a\leq x \leq b\}$.

A graph $G$ is an \emph{interval graph} if there exists a bijection from the vertices of $G$ to a multiset of intervals, $f:V\to \mathcal{I}$, such that there exists an edge between two vertices if and only if their corresponding intervals intersect. The multiset $\mathcal{I}$ is called an \emph{interval representation} of $G$.

For any natural number $d>0$, a (disjoint) \emph{$d$-interval} is the union of $d$ (disjoint) intervals on the real line. 
    
For any natural number $d>0$, a graph $G$ is a (disjoint) \emph{$d$-interval graph} if there exists a bijection from the vertices of $G$ to a multiset of (disjoint) $d$-intervals, $f:V\to \mathcal{I}$, such that there exists an edge between two vertices if and only if their corresponding $d$-intervals intersect. The multiset $\mathcal{I}$ of $d$-intervals is called a $d$-interval representation of $G$, and the family of all intervals that compose the $d$-intervals in $\mathcal{I}$ is called the \emph{underlying family of intervals} of $\mathcal{I}$.

A (disjoint) $d$-interval graph is \emph{unit} if there exists a (disjoint) $d$-interval representation where all the intervals of the underlying family have unit length, and it is \emph{proper} if there exists a representation where no interval of the underlying family is properly contained in another one.
A (disjoint) $d$-interval graph is \emph{balanced} if there exists a (disjoint) $d$-interval representation where the $d$ intervals of a same $d$-interval have the same length, but intervals of different $d$-intervals can differ in length. 

The graph $K_{1,t}$ is the star with $t$ leaves (also referred to as \emph{$t$-claw} in the following).
For any $t \geq 3$, if the set of vertices $\{v_0, v_1, \ldots, v_t\}$ induces a $K_{1,t}$ with center $v_0$, we will denote it by $[v_0; v_1, \ldots , v_t]$. We say that a graph is $K_{1,t}$-free if it does not contain any induced $K_{1,t}$'s.
Furthermore, we say that an induced $t$-claw $K_{1,t}$ is \emph{maximal} if it is not contained in an induced $K_{1,m}$ with $m>t$.

\subparagraph*{Discussion on the definition of $d$-intervals.}

As mentioned in the introduction, $d$-interval graphs have been defined in the literature both as the union of $d$ disjoint intervals and as the union of $d$ not necessarily disjoint intervals. 
This might be related to the fact that when there are no length restrictions on the intervals, both definitions lead to the same class of graphs (as one can simply stretch the intervals associated to a same vertex that intersect to make them disjoint without changing any of the other intersections).

\begin{observationrep}\label{obs:disjoint free equivalent}
    The classes of disjoint $d$-interval and $d$-interval graphs are equivalent.
\end{observationrep}
\begin{proof}
It is clear that the class of disjoint $d$-interval graphs is contained in the class of $d$-interval graphs. To see the other direction, it suffices to notice that if we have a $d$-interval representation, one could represent every pair of intersecting intervals $[a,b]$ and $[c,d]$ (with $a<c<b<d$), associated to the same vertex $v$, by a single interval $[a,d]$ to obtain an equivalent disjoint $d$-interval representation (note that since we have decreased the number of intervals associated to vertex $v$ by one, to obtain a $d$-interval representation, one would technically need to add a ``dummy'' interval associated to $v$ which does not intersect any other interval in the representation). 
 
\end{proof}

However, if there are length restrictions, the previous observation does not hold. For unit intervals, one cannot replace two intersecting intervals $[a,b]$ and $[c,d]$, with $a<c<b<d$, by $[a,d]$, as the resulting interval would not be of unit length, and stretching it to make it unit might disrupt the rest of the intersections. 
Thus, in this case, it cannot be inferred that both definitions of multiple intervals lead to the same class of graphs. In fact, our results prove that they do not. 
Therefore, we study the generalization of Roberts characterization separately for both definitions of $d$-intervals. 





\section{Unit d-interval graphs}\label{SectionAlgorithm}

In this section, we generalize Roberts characterization of unit interval graphs for $d$-interval graphs. Recall that by $d$-interval graphs we refer to intersection graphs of $d$-intervals where the $d$ intervals are not necessarily disjoint, or in other words, to the most general definition.

\begin{theorem}\label{k12d+1-free are unit}
    Let $G$ be an interval graph. Then, for any natural number $d\geq 2$, $G$ is a unit $d$-interval graph if and only if $G$ does not contain a copy of a $K_{1,2d+1}$ as an induced subgraph. Furthermore, given a $K_{1,2d+1}$-free interval graph, a unit $d$-interval representation can be constructed in $\mathcal{O} (n +m)$ time, where $n$ and $m$ are the number of vertices and edges of the graph, respectively.
\end{theorem}

We present a polynomial-time algorithm that, given an arbitrary interval representation $\mathcal{I}$ of a $K_{1,2d+1}$-free graph, returns a $d$-interval representation $\mathcal{I'}$ of the graph where no interval of the underlying family of $\mathcal{I'}$ intersects three or more pairwise disjoint intervals. This ensures that the underlying family of intervals returned corresponds to an interval representation of a $K_{1,3}$-free graph, so we can use the algorithm described in~\cite{DBLP:journals/dm/BogartW99} to turn it into a proper representation (and then to a unit one in linear time \cite{DBLP:journals/dm/Gardi07}). 
Note that if an interval representation of the graph is not given, we can always compute it in linear time~\cite{DBLP:journals/siamdm/CorneilOS09}.

Before presenting the algorithm formally, let us give the idea behind it.
The algorithm constructs a family $\mathcal{I'}$ of $d$-intervals in the following way: for every interval $I\in\mathcal{I}$ that intersects $m$ (and no more than $m$) pairwise disjoint intervals, we create a $t$-interval $I_1 \cup \ldots \cup I_t$, where $t=\lceil \frac{m}{2} \rceil$. Note that for every interval $I$ that intersect only two disjoint intervals, we have $t=1$, and the interval $I_1$ added to $\mathcal{I'}$ will be exactly $I$. We will refer to such intervals as \emph{original intervals}, as they are equal to the ones in $\mathcal{I}$.
After creating the $t$-intervals described above, to obtain a $d$-interval representation of the graph, it suffices to add $d-t$ ``dummy'' intervals for each vertex that is represented by $t<d$ intervals (where by ``dummy'' intervals we mean that they do not intersect any other interval from the representation). Each $d$-interval $I_1 \cup \ldots \cup I_d$ introduced will preserve the same intersections as the interval $I\in \mathcal{I}$, and each $I_i$ will possess three key properties: it intersects at most two disjoint original intervals, it contains an original interval, and each of its endpoints coincides with an endpoint of an original interval. These properties ensure that the representation $\mathcal{I'}$ can be made unit. 

\paragraph*{Algorithm}
Let the family of intervals $\mathcal{I}$ be an interval representation of $G$. For every interval $I \in \mathcal{I}$, let $l(I)$ and $r(I)$ stand for its left and right endpoint, respectively.
Furthermore, define a partial order as follows: given two intervals $I,J \in \I$, let $I \prec J$ if and only if $r(I)<l(J)$ (i.e. interval $J$ is fully to the right of interval $I$). Two intervals are incomparable if they intersect. 
\begin{enumerate}
    \item[\textbf{Step 1}]\label{algo:step 0} Initialize a set of intervals $\mathcal{C}$ with all the intervals of $\mathcal{I}$, set $\I':=\emptyset$, 
    and go to \ref{algo:step 1}.
    \item[\textbf{Step 2}] \label{algo:step 1} Pick an interval $I$ of $\mathcal{C}$, remove it from the set and define its neighborhood $\mathcal{N}(I)=\{J \in \mathcal{I} \colon J \cap I \neq \emptyset \}$. Let $m$ be the maximum number of pairwise disjoint intervals that $I$ intersects. If $m\leq 2$, go to \ref{algo:step 2}; if $m=3$, go to  \ref{algo:step 3}; and if $m>3$, go to \ref{algo:step 4}. 
    \item[\textbf{Step 3}]\label{algo:step 2} If $m\leq 2$, add the interval $I_1=I$ to the family $\mathcal{I'}$ and call $I_1$ an original interval. Then go to \ref{algo:step 5}.
    \item[\textbf{Step 4}]\label{algo:step 3} If $m=3$, define four auxiliary intervals: 
    \begin{align*}
        A_1 &= \argmin_{J \in \mathcal{N}(I)} \{r(J)\} \quad \quad
            &A_2 &= \argmin_{\{J \in \mathcal{N}(I)\colon A_{1}\prec J\}} \{r(J)\} \\
            A_{4} &= \argmax_{J\in \mathcal{N}(I)} \{l(J)\} \quad \quad
            &A_{3} &= \argmax_{\{J\in \mathcal{N}(I) \colon J\prec A_{4}\}} \{l(J)\}
    \end{align*}
     Then add to $\mathcal{I'}$ the 2-interval $I_1\cup I_2$, with $I_1 = [l(I), r(A_2)]$ and $I_2= [l(A_3), r(I)]$. Note that $A_2$ and $A_3$ necessarily intersect, as otherwise we would have $m\geq 4$, so $I_1\cup I_2$ is not a disjoint 2-interval. After adding it to $\mathcal{I'}$, go to \ref{algo:step 5}.
    \item[\textbf{Step 5}]\label{algo:step 4} If $m>3$, define two families of auxiliary intervals. The first family $\mathcal{A}:=\{A_i \,|\, i \in \{1,\ldots, m\}\}$ forms a maximum set of pairwise disjoint intervals intersecting $I$, and it will ensure that all the intersections are preserved. It is defined as follows:
    \begin{align*}
            A_1 &= \argmin_{J \in \mathcal{N}(I)} \{r(J)\} 
            &A_i &= \argmin_{\{J \in \mathcal{N}(I)\colon A_{i-1}\prec J\}} \{r(J)\} \text{  , $\forall$ } i\in\{2,\ldots,m-2\} \\
            A_{m} &= \argmax_{J\in \mathcal{N}(I)} \{l(J)\} 
            &A_{m-1} &= \argmax_{\{J\in \mathcal{N}(I) \colon J\prec A_{m}\}} \{l(J)\}
    \end{align*}
        The second family $\mathcal{B}:=\{B_i\, |\, i \in \{1,\ldots, m\}\}$ is a tool to ensure that each new interval $I_i$ intersects only two disjoint intervals in $\mathcal{I'}$. Note that restricting each $I_i$ to intersect only two disjoint intervals from the family $\mathcal{A}$ is not enough: for example, in \cref{fig:alg}, if $I_2$ were defined as $[l(A_3), r(A_4)]$, then it would intersect three pairwise disjoint intervals in $\mathcal{I'}$ (as all the intervals except $I$ are original intervals in this example), whereas if the left endpoint of $I_2$ were chosen as $r(A_{2i-1})$, then an original interval that was not the center of a claw in $\mathcal{I}$ might become the center of a new claw in $\mathcal{I'}$. 
        Thus, for every $i\in\{1,\ldots, m\}$, $B_i$ is defined as follows:
            \begin{align*}  
            B_{i} &= \argmax_{J\in \mathcal{N}(A_{i}
        )\cup A_{i}} \{l(J)\}
            \end{align*}
        In other words, $B_{i}$ is the interval in the closed neighborhood of $A_{i}$ starting the latest. Note that if there does not exist any interval intersecting $A_{i}$ which starts after $A_{i}$, then $B_{i}=A_{i}$ since we are considering the closed neighborhood. 
    Now, add to $\mathcal{I'}$ the $t$-interval $I_1\cup ...\cup I_t$, defined as follows. We distinguish two slightly different cases:
    \begin{enumerate}
        \item  If $m$ is even, i.e., $m=2t$ for some $t>1$, define $I_1= [l(I), r(A_2)]$, $I_i = [l(B_{2i-1}), r(A_{2i})]$ for every $i \in\{2,\ldots,t-1\}$, and 
        $I_t = [l(A_{2t-1}, r(I)]$.
       
        \item If $m$ is odd, i.e., $m=2t-1$ for $t>2$, define $I_{t-1}$ and $I_{t}$ differently, as $I_{t-1} = A_{2t-3}$ and $I_t = [l(A_{2t-2}, r(I)]$, and the rest of the intervals as before.
       
    \end{enumerate}
Notice that by definition, the intervals $I_1, ..., I_t$ are actually pairwise disjoint, so if $m>3$, the $t$-interval added to $\mathcal{I'}$ is a disjoint $d$-interval. After adding the $t$-interval, go to \ref{algo:step 5}.
\item\label{algo:step 5} If $\mathcal{C}=\emptyset$, return $\mathcal{I'}$, else go to \ref{algo:step 1}.
\end{enumerate}

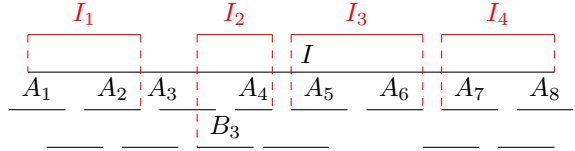
\begin{figure}
    \centering
      \centering
        
    \begin{tikzpicture}[scale=0.25]

    \draw (-10,5) -- (18,5) node[midway,above right] {$I$};

\draw[red] (-10,7) -- (-4,7) node[midway,above] {$I_1$};
\draw[red] (-1,7) -- (3,7) node[midway,above] {$I_2$};
\draw[red] (4,7) -- (11,7) node[midway,above] {$I_3$};
\draw[red] (12,7) -- (18,7) node[midway,above] {$I_4$};

\draw[red, dashed] (-10,7) -- (-10,5) ;
\draw[red, dashed] (-4,7) -- (-4,3) ;
\draw[red, dashed] (-1,7) -- (-1,1) ;
\draw[red, dashed] (3,7) -- (3,3) ;
\draw[red, dashed] (4,7) -- (4,3) ;
\draw[red, dashed] (11,7) -- (11,3) ;
\draw[red, dashed] (12,7) -- (12,3) ;
\draw[red, dashed] (18,7) -- (18,5) ;

    \draw (-11,3) -- (-8,3) node[midway,above] {$A_1$};
    \draw (-7,3) -- (-4,3) node[midway,above] {$A_2$};
    \draw (-3,3) -- (0,3) node[midway,above left= 0.0001cm] {$A_3$};
    \draw (1,3) -- (3,3) node[midway,above] {$A_4$};
    \draw (4,3) -- (7,3) node[midway,above] {$A_5$};
    \draw (8,3) -- (11,3) node[midway,above] {$A_6$};
    \draw (12,3) -- (15,3) node[midway,above] {$A_7$};
    \draw (16,3) -- (19,3) node[midway,above] {$A_8$};

    \draw (-9,1) -- (-6,1) node[midway,above] {};
    \draw (-5,1) -- (-2,1) node[midway,above] {};
    \draw (-1,1) -- (2,1) node[midway,above] {$B_3$};
    \draw (2.5,1) -- (6,1) node[midway,above] {};
    \draw (11,1) -- (14,1) node[midway,above] {};
    \draw (15,1) -- (18,1) node[midway,above] {};


    \end{tikzpicture}
    \caption{Interval $I$ intersects 8 disjoint intervals. In red, the 4-interval returned by the algorithm. Note that if $l(I_2)$ were defined as $l(A_3)$ instead of $l(B_3)$, it would create a forbidden $K_{1,3}$.}
    \label{fig:alg}
\end{figure}

\cref{fig:alg} illustrates the algorithm on a concrete interval which intersects eight pairwise disjoint intervals. Before proceeding to the proof of correctness of the algorithm, we highlight the properties of the intervals constructed that will be useful to prove the next three claims. In the following, we say that an interval $I$ of $\mathcal{I}$ has been \emph{transformed} into a $t$-interval $I_1 \cup \dots \cup I_t$ by the algorithm after it has been processed.

\begin{observation}\label{claim:properties}
    Let $I\in \mathcal{I}$ be an interval transformed into $I_1 \cup \dots \cup I_t$ by the algorithm, for some $1< t\leq d$.
    Then, for every $i\in\{1,\ldots, t\}$:
    \begin{enumerate}
        \item The left (resp., right) endpoint of every interval $I_i$ coincides with the left (resp., right) endpoint of an original interval.
        \item There is an original interval contained in $I_i$.  
    \end{enumerate}
\end{observation}

Now, to prove the correctness of the algorithm, we need to show that for every interval $I\in \mathcal{I}$, the $t$-interval $I_i\cup ...\cup I_t \in \mathcal{I'}$ preserves the same intersections as $I$, and that no interval in the underlying family of $\mathcal{I'}$ intersects three pairwise disjoint intervals. In the next claim, we prove that intersections are preserved:

\begin{claimrep}\label{claim:preserve intersections}
     Let $I$ be an interval transformed into $I_1 \cup \dots \cup I_t$ by the algorithm, for some $1\leq t\leq d$.
    Then, the $t$-interval $I_1 \cup \ldots \cup I_t$ preserves the intersections of $I$.
\end{claimrep}
\begin{claimproof}
    It is clear that no new intersections are created as $I_1 \cup \ldots \cup I_t \subseteq I$. 
    To see that no intersection is lost, suppose that there exists an interval $L$ that intersects $I$ in the original representation $\mathcal{I}$, and after the algorithm finishes, $L$ is transformed into a $t_0$-interval $L_1 \cup \ldots \cup L_{t_0}$ (for some $1\leq t_0\leq d$, where if $t_0=1$, the interval remains as in the original representation) such that
    the $t$-interval $I_1 \cup \ldots \cup I_t$ does not intersect the $t_0$-interval $L_1 \cup \ldots \cup L_{t_0}$ in $\mathcal{I'}$.
    Since $l(I_1)=l(I)$ and $r(I_t)=r(I)$ (and the same holds for $L$), this means that there exists an $L_j$ (with $1\leq j \leq t_0$) such that $I_i \prec L_j \prec I_{i+1}$ for some $1\leq i \leq t-1$. 
    
    For $1\leq t\leq 2$, this cannot occur because $I\subseteq I_1 \cup \ldots \cup I_t $.
    For $t>3$, since the set of intervals $A_k$ used to defined the $t$-interval associated to $I$ forms a maximal set of pairwise disjoint intervals intersecting $I$, 
    we cannot have that $A_{k} \prec L_j \prec A_{k+1}$ for any $1\leq k \leq 2t-1$. Indeed, this would contradict maximality, as $L_j$ is either an original interval or it contains an original interval (by \cref{claim:properties}).
    Thus, the only possible option is that there exists an $i$ such that $A_{2i}\prec L_j \prec B_{2i+1}$ (where $B_{2i+1}$ is different from $A_{2i+1}$). Then, since $B_{2i+1}$ intersects $A_{2i+1}$ and $L_j \prec B_{2i+1}$, we have that $r(L_j) < r(A_{2i+1})$. But this contradicts the choice of $A_{2i+1}$, which should have been $L_j$ or the original interval contained in $L_j$, as $A_{2i}\prec L_j$.
\end{claimproof}

The next two claims are dedicated to proving that no interval in the underlying family of $\mathcal{I'}$ intersects three or more pairwise disjoint intervals. We distinguish the cases when the center of the claw is an original interval and when it is not.

\begin{claim}\label{claim:unit}
Let $I\in\mathcal{I}$ be an original interval (i.e., transformed to $I_1$ by the algorithm). Then, $I_1$ intersects at most two disjoint intervals in the underlying family of $\mathcal{I'}$.
\end{claim}
\begin{claimproof}
    Suppose, towards a contradiction, that there exists an original interval $I_1$ that intersects three pairwise disjoint intervals $L_1, L_2$ and $L_3$ in the underlying family of $\mathcal{I'}$, with $L_1 \prec L_2 \prec L_3$. 
    By \cref{claim:properties}, there exists an original interval $L_1'$ with the same right endpoint as $L_1$, an original interval $L_2'$ contained in $L_2$, and an original interval $L_3'$ with the same left endpoint as $L_3$. 
    Note that if any of the $L_i$ are original, then $L_i'=L_i$. But then, $L_1' \prec L_2' \prec L_3'$ are three pairwise disjoint original intervals that intersect $I_1$, which contradicts the fact that it is an original interval. Indeed, this implies that the interval $I$ intersects three pairwise disjoint intervals in $\mathcal{I}$, and so the algorithm would have transformed it into a $t$-interval with $t$ strictly greater than 1. 
 \end{claimproof}

    \begin{claim}\label{partialinterval}
     Let $I\in\mathcal{I}$ be an interval transformed into the $t$-interval $I_1 \cup \dots \cup I_t$ by the algorithm, for some $1<t\leq d$.
    For every $1\leq i \leq t$, $I_i$ intersects at most two disjoint intervals of the underlying family of $\mathcal{I'}$.
    \end{claim}
    \begin{claimproof}
        We proceed by contradiction. Suppose that there exists an interval $I_i$, with $1\leq i \leq t$ that intersects three pairwise disjoint intervals $L_1, L_2, L_3$, with $L_1 \prec L_2 \prec L_3$. By \cref{claim:properties}, there exists an original interval $L_1'$ with the same right endpoint as $L_1$, an original interval $L_2'$ contained in $L_2$, and an original interval $L_3'$ with the same left endpoint as $L_3$.
    
        Assume first that $t=2$. Then, if $i=1$, this contradicts the choice of the interval $A_2$ (resp. $A_3$ if $i=2$), which should have been $L_2'$. 

        Let us now study the general case for $t>2$. Suppose first that $1<t<d$ and $I_i$ is defined as $[B_{2i-1}, A_{2i}]$ with $B_{2i-1} \neq A_{2i-1}$. 
        We distinguish two cases:
        \begin{enumerate}
            \item $r(L_1) > r(A_{2i-1})$. Then, since we are assuming that $L_1$ and $L_2$ are disjoint, $l(L_2)> r(A_{2i-1})$. Furthermore, as $L_3$ also intersects $I_i$, we need $r(L_2) < r(A_{2i})$. But this contradicts the choice of $A_{2i}$, which should have been $L_2'$. 
            \item $r(L_1)< r(A_{2i-1})$. If $l(L_2) > r(A_{2i-1})$, we are in the same case as before. Thus, $L_2$ and $A_{2i-1}$ must intersect. However, we have $l(L_2)> l(B_{2i-1})$ (since otherwise $I_i$ would not be able to intersect $L_1$ on its left extreme). This contradicts the choice of $B_{2i-1}$ if $L_2'$ intersects $A_{2i-1}$, or the choice of $A_{2i}$ otherwise. 
        \end{enumerate}
       On the other hand, if $B_{2i-1}=A_{2i-1}$, then by construction, since we take the two disjoint intervals that finish first, we cannot have three pairwise disjoint intervals intersecting $I_i$. This is also the case for $I_1$ and $I_t$ (although in the latter case, we take the two disjoint intervals starting last). Finally, for odd claws, it is also clear that $I_{t-1}$ intersects at most two disjoint intervals, as it is equal to an original interval.
     \end{claimproof}

Combining Claims~\ref{claim:preserve intersections}, \ref{claim:unit} and \ref{partialinterval}, plus the fact that we can trivially transform a $t$-interval with $t<d$ into a $d$-interval, we obtain that the algorithm returns a $d$-interval representation of the input graph where no interval of the underlying family intersects more than two disjoint intervals, which as explained before can be converted into a unit representation. 
The last part of \cref{k12d+1-free are unit} follows because an efficient implementation of the algorithm described above requires $\mathcal{O}(1+deg(v))$ operations for each vertex $v$ (where $deg(v)$ denotes the degree of vertex $v$), as it suffices to iterate over the neighborhood of a given interval to transform it into the corresponding $d$-interval. Finally, the obtained representation can be converted to a unit representation in linear time, which yields the stated runtime $\mathcal{O}(n + m)$.
This concludes the proof of \cref{k12d+1-free are unit}.   

We have proven that the algorithm constructs a unit $d$-interval representation, but it is not a disjoint one. Indeed, as mentioned before, in the case of maximal $K_{1,3}$'s, the constructed intervals $I_1$ and $I_2$ intersect each other. However, in the case of maximal $K_{1,m}$'s with $m>3$, the $t$ intervals of the $t$-interval created are actually pairwise disjoint. Thus, we obtain as a direct corollary that if $G$ is a $K_{1,2d+1}$-free interval graph not containing any maximal $K_{1,3}$'s, then $G$ is a disjoint unit $d$-interval graph. 
In fact, with a more careful analysis, we can infer an even stronger corollary, which instead of requiring the absence of maximal 3-claws all together, only forbids a subset of them. We refer to these forbidden claws, which are exactly those maximal 3-claws contained in an induced $E$ graph, as $E$-claws. Recall that an $E$ graph (or $\text{star}_{1,2,2}$) is a graph on six vertices which has as edge set a path $v_1,v_2,v_3,v_4,v_5$ and an additional edge $(v_3,v_6)$.


\begin{corollaryrep}\label{corollary:E claws}
    Let G be a $K_{1,2d+1}$-free graph that does not contain any $E$-claws. Then, $G$ is a disjoint unit $d$-interval graph.
\end{corollaryrep}

\begin{proof}
To prove the theorem, we modify \ref{algo:step 3} of the previous algorithm so that it produces a disjoint 2-interval. 

\begin{enumerate}
    \item[\textbf{Step 4'}]\label{algo:step 4prime} Let $I$ be an interval and let $m=3$ be the maximum number of pairwise disjoint intervals that it intersects. By assumption, the vertex associated to $I$ is a center of a maximal claw which is not an $E$-claw. 
We define  \begin{align*}
            A_1 &= \argmin_{J \in \mathcal{N}(I)} \{r(J)\} \\
            A_2 &= \argmin_{\{J \in \mathcal{N}(I)\colon A_{1}\prec J\}} \{r(J)\} \\
            A_{4} &= \argmax_{J\in \mathcal{N}(I)} \{l(J)\} \\
            A_{3} &= \argmax_{\{J\in \mathcal{N}(I) \colon J\prec A_{4}\}} \{l(J)\}
    \end{align*}
Note that $A_2$ and $A_3$ necessarily intersect (or are the same interval). Now, since the vertex associated to $I$ is a center of a claw that is not an $E$-claw, this means that at least one of $A_1$ or $A_4$ does not intersect an interval which is disjoint from $I$.
Thus, we can modify the representation so that $A_1$ (resp. $A_4$) is properly contained in $I$ without loosing any intersections, by simply stretching them. 
Then, if $A_1$ is properly contained in $I$, we define $I_1 = A_1$ and $I_2 = [l(A_3), r(I)]$. On the other hand, if $A_4$ is properly contained in $I$ instead, we define $I_1=[l(I), r(A_2)]$ and $I_2=A_4$. If both of them are properly contained in $I$, we can define $I_1$ and $I_2$ either way. 
\end{enumerate}
Notice that the 2-intervals introduced in this step have the same properties as in \cref{claim:properties}, so the proof of correctness of the previous algorithm can be directly adapted for this extension. 
\end{proof}

\section{Disjoint unit d-interval graphs}\label{SectionGeneral}

In this section, we prove that \cref{k12d+1-free are unit} cannot be generalized for disjoint unit $d$-interval graphs. Note that by \cref{corollary:E claws}, if we have a graph which does not contain any $E$-claws
, then the generalization still holds for disjoint unit $d$-interval graphs, but this is not the case in general.
Indeed, suppose there is an interval $I$ that intersects exactly three pairwise disjoint intervals, $A_1$, $A_2$ and $A_3$, and both $A_1$ and $A_3$ intersect each an interval disjoint from $I$. Then, the algorithm presented in the previous section would return a 2-interval $I_1 \cup I_2$, where $I_1$ and $I_2$ are not disjoint. If we try to extend the algorithm in the most natural way, that is, stretching these two intervals until they are disjoint, we would still not succeed. This is because, since $I_1$ and $I_2$ cannot intersect, then either $r(I_1)$ will be to the left of the right endpoint returned by the algorithm, or $l(I_2)$ will be to the right of the endpoint returned by the algorithm.
But then, one of $I_1$ or $I_2$ might not properly contain a complete interval from the original representation, which can cause $I_1$ or $I_2$ -- the interval which does not properly contain a complete original interval -- to be contained in an interval that intersects three pairwise disjoint intervals (see \cref{fig:max3claw}). 

\begin{figure}
    \centering
        \centering
        
    \begin{tikzpicture}[scale=0.4]

    \draw (-3,7) -- (3,7)  node[midway,above]{\scriptsize $8$};
    \draw (-10,5) -- (-2,5)  node[midway,above left] {\scriptsize $1$};
    \draw (-1,5) -- (1,5) node[midway,above] {\scriptsize $9$};
    \draw (2,5) -- (10,5) node[midway,above right] {\scriptsize $10$};
    \draw (-10,3) -- (-9,3) node[midway,above] {\scriptsize $2$};
    \draw (-8,3) -- (-7,3) node[midway,above] {\scriptsize $4$};
    \draw (-6,3) -- (-5,3) node[midway,above] {\scriptsize $6$};
    \draw (-6,4) -- (-4,4) node[midway,above] {\scriptsize $5$};
 \draw (9,3) -- (10,3) node[midway,above] {\scriptsize $14$};
    \draw (7,3) -- (8,3) node[midway,above] {\scriptsize $13$};
    \draw (5,3) -- (6,3) node[midway,above] {\scriptsize $12$};
    \draw (4,4) -- (6,4) node[midway,above] {\scriptsize $11$};

    \draw (-8,1) -- (4.5,1) node[midway,above left] {\scriptsize $3$};
    \draw (-4.5,0) -- (8,0) node[midway,above] {\scriptsize $7$};

    \draw[red] (-8,1.5) -- (-5,1.5) node[midway,above] {\scriptsize $3_1$};
        \draw[red] (-1,1.5) -- (4.5,1.5) node[midway,above right= 0.1 cm] {\scriptsize $3_2$};

    \draw[red] (-4.5,-1) -- (1,-1) node[midway,above] {\scriptsize $7_1$};
        \draw[red] (5,-1) -- (8,-1) node[midway,above] {\scriptsize $7_2$};
        
           \draw[red] (-3,8) -- (1,8)  node[midway,above]{\scriptsize $8_1$};
        \draw[red] (2,8) -- (3,8)  node[midway,above]{\scriptsize $8_2$};

    \draw[dashed] (3, 8) -- (3, 1.5);
    \draw[dashed] (1, 8) -- (1, 1.5);
    \draw[dashed] (4, 4) -- (4, 1.5);
    \draw[dashed] (2, 8) -- (2, 1.5);

     \draw[red] (-10,6) -- (-7,6)  node[midway,above] {\scriptsize $1_1$};
     \draw[red] (-6,6) -- (-2,6)  node[midway,above] {\scriptsize $1_2$};
    \draw[red] (2,6) -- (6,6) node[midway,above] {\scriptsize $10_1$};
        \draw[red] (7,6) -- (10,6) node[midway,above] {\scriptsize $10_2$};

    \end{tikzpicture}
    \caption{Interval representation of a $K_{1,5}$-free graph that cannot be turned into a disjoint unit 2-interval representation just by ``cutting'' intervals that intersect more than three pairwise disjoint intervals. In the figure, the intervals in red are all obtained using a natural extension of the algorithm. We can see that in this way, $3_2$ intersects three disjoint intervals: $8_1, 8_2, 11$. The reader can check that no other way of stretching the intervals works if $8_1$ and $8_2$ are required to be disjoint.}
    
    \label{fig:max3claw}
\end{figure}

In the following, we show that there is no way to extend the algorithm to make it work in the general case for disjoint unit $d$-interval graphs. In particular, we prove the following theorem.

\begin{theorem}\label{conj:roberts}
    There exists a $K_{1,5}$-free interval graph that is not a disjoint unit 2-interval graph.
\end{theorem}


To prove \cref{conj:roberts}, we offer the graph $G$ in~\cref{fig:counterexample} as a certificate.
The reader can check that $G$ has no induced $K_{1,5}$, and
an interval representation of $G$ is provided in~\cref{fig:representation}. 
The proof that $G$ is not a disjoint unit 2-interval graph is the challenging part. 
Indeed, checking whether a graph is disjoint unit 2-interval is a computationally expensive task, and even with the aid of computer search, a naive ILP implementation already takes too much time to consider an exhaustive search. 
Needless to say, checking manually by brute force leads to a very long branching process. The proof presented here is based on a careful analysis of the graph, and the technique employed (which uses the characterization of unit 2-interval graphs in [\cite{acceptedvirginia}, Lemma 5]) may be applied to establish that other graphs are not disjoint unit 2-interval graphs. We also verify the proof computationally, using an encoding in answer set programming based on the semiorder characterization of unit interval graphs, which proves to be way more efficient than an ILP encoding. Our code and experimental setting can be found on our git repository~\footnote{\url{https://github.com/AbdallahS/unit-graphs}}.
Furthermore, we provide five other $K_{1,5}$-free interval graphs on the same number of vertices (and with a very similar structure) that are not disjoint unit 2-interval (see \cref{pics:5counterexamples}). The proof that they are not disjoint unit 2-interval is omitted, but it is analogous to the one presented here. These six graphs are the only such graphs 
on 14 vertices, and there does not exist a graph satisfying the conditions of \cref{conj:roberts} with fewer vertices. These assertions were verified by computer search over all interval graphs of a given size without induced $K_{1,5}$'s~\cite{DBLP:journals/tcs/YamazakiSKU20}. 


\begin{toappendix}  


\begin{figure}[htp]
\begin{subfigure}[b]{0.5\textwidth}
\centering
   \resizebox{6cm}{3cm}{%

\begin{tikzpicture}[ node distance={15mm}, main/.style = {draw, circle}] 

\node[main] (8) {$8$}; 

\node[main] (9) [below of=8] {$9$}; 
\node[main] (1) [left of=9] {$1$};
\node[main] (10) [ right of=9] {$10$};
\node[main] (4) [below left of=1] {$4$}; 
\node[main] (13) [below right of=10] {$13$};
\node[main] (3) [below of=1] {$3$}; 
\node[main] (7) [below of=10] {$7$};
\node[main] (14) [ right of=13] {$14$};
\node[main] (12) [below of=14] {$12$};
\node[main] (11) [below of=13] {$11$};

\node[main] (2) [ left of=4] {$2$};
\node[main] (6) [below of=2] {$6$};
\node[main] (5) [below of=4] {$5$};

\draw (8) -- (9);
\draw (8) -- (1);
\draw (8) -- (10);
\draw (8) -- (3);
\draw (8) -- (7);
\draw (9) -- (3);
\draw (9) -- (7);
\draw (1) -- (2);
\draw (1) -- (3);
\draw (1) -- (4);
\draw (1) -- (7);
\draw (10) -- (3);
\draw (10) -- (7);
\draw (10) -- (13);
\draw (10) -- (11);
\draw (10)[out=330, in=110] to (12);
\draw (10) -- (14);
\draw (3) -- (7);
\draw (3) -- (5);
\draw (3) -- (6);
\draw (5) -- (6);

\draw (7) -- (11);
\draw (7) -- (12);
\draw (11) -- (12);
\draw (3) -- (4);
\draw (5) -- (7);
\draw (1) -- (5);
\draw (1)[out=210, in=70] to (6);
\draw (7) -- (13);
\draw (3) -- (11);

\draw[red, dashed, very thick] (11) -- (13);
\draw[red, dashed, very thick] (11) -- (14);
\draw[red, dashed, very thick] (5) -- (4);
\draw[red, dashed, very thick]  (5) -- (2);

\end{tikzpicture} 
   
}
\end{subfigure}
\quad
\begin{subfigure}[b]{0.5\textwidth}
\centering
\resizebox{6cm}{3cm}{%

\begin{tikzpicture}[node distance={15mm}, main/.style = {draw, circle}] 

\node[main] (8) {$8$}; 

\node[main] (9) [below of=8] {$9$}; 
\node[main] (1) [left of=9] {$1$};
\node[main] (10) [ right of=9] {$10$};
\node[main] (4) [below left of=1] {$4$}; 
\node[main] (13) [below right of=10] {$13$};
\node[main] (3) [below of=1] {$3$}; 
\node[main] (7) [below of=10] {$7$};
\node[main] (14) [ right of=13] {$14$};
\node[main] (12) [below of=14] {$12$};
\node[main] (11) [below of=13] {$11$};

\node[main] (2) [ left of=4] {$2$};
\node[main] (6) [below of=2] {$6$};
\node[main] (5) [below of=4] {$5$};

\draw (8) -- (9);
\draw (8) -- (1);
\draw (8) -- (10);
\draw (8) -- (3);
\draw (8) -- (7);
\draw (9) -- (3);
\draw (9) -- (7);
\draw (1) -- (2);
\draw (1) -- (3);
\draw (1) -- (4);
\draw (1) -- (7);
\draw (10) -- (3);
\draw (10) -- (7);
\draw (10) -- (13);
\draw (10) -- (11);
\draw (10)[out=330, in=110] to (12);
\draw (10) -- (14);
\draw (3) -- (7);
\draw (3) -- (5);
\draw (3) -- (6);
\draw (5) -- (6);

\draw (7) -- (11);
\draw (7) -- (12);
\draw (11) -- (12);
\draw (3) -- (4);
\draw (5) -- (7);
\draw (1) -- (5);
\draw (1)[out=210, in=70] to (6);
\draw (7) -- (13);
\draw (3) -- (11);

\draw[red, dashed, very thick] (11) -- (13);
\draw[red, dashed, very thick] (5) -- (4);
\draw[red, dashed, very thick]  (5) -- (2);

\end{tikzpicture} 
   
}
\end{subfigure}

\medskip
\begin{subfigure}[b]{0.5\textwidth}
\resizebox{6cm}{3cm}{%
\begin{tikzpicture}[node distance={15mm}, main/.style = {draw, circle}] 

\node[main] (8) {$8$}; 

\node[main] (9) [below of=8] {$9$}; 
\node[main] (1) [left of=9] {$1$};
\node[main] (10) [ right of=9] {$10$};
\node[main] (4) [below left of=1] {$4$}; 
\node[main] (13) [below right of=10] {$13$};
\node[main] (3) [below of=1] {$3$}; 
\node[main] (7) [below of=10] {$7$};
\node[main] (14) [ right of=13] {$14$};
\node[main] (12) [below of=14] {$12$};
\node[main] (11) [below of=13] {$11$};

\node[main] (2) [ left of=4] {$2$};
\node[main] (6) [below of=2] {$6$};
\node[main] (5) [below of=4] {$5$};

\draw (8) -- (9);
\draw (8) -- (1);
\draw (8) -- (10);
\draw (8) -- (3);
\draw (8) -- (7);
\draw (9) -- (3);
\draw (9) -- (7);
\draw (1) -- (2);
\draw (1) -- (3);
\draw (1) -- (4);
\draw (1) -- (7);
\draw (10) -- (3);
\draw (10) -- (7);
\draw (10) -- (13);
\draw (10) -- (11);
\draw (10)[out=330, in=110] to (12);
\draw (10) -- (14);
\draw (3) -- (7);
\draw (3) -- (5);
\draw (3) -- (6);
\draw (5) -- (6);

\draw (7) -- (11);
\draw (7) -- (12);
\draw (11) -- (12);
\draw (3) -- (4);
\draw (5) -- (7);
\draw (1) -- (5);
\draw (1)[out=210, in=70] to (6);
\draw (7) -- (13);
\draw (3) -- (11);

\draw[red, dashed, very thick] (5) -- (4);
\draw[red, dashed, very thick]  (5) -- (2);

\end{tikzpicture} 
   
}
\end{subfigure}
\quad
\begin{subfigure}[b]{0.5\textwidth}
\resizebox{6cm}{3cm}{%

\begin{tikzpicture}[node distance={15mm}, main/.style = {draw, circle}] 

\node[main] (8) {$8$}; 

\node[main] (9) [below of=8] {$9$}; 
\node[main] (1) [left of=9] {$1$};
\node[main] (10) [ right of=9] {$10$};
\node[main] (4) [below left of=1] {$4$}; 
\node[main] (13) [below right of=10] {$13$};
\node[main] (3) [below of=1] {$3$}; 
\node[main] (7) [below of=10] {$7$};
\node[main] (14) [ right of=13] {$14$};
\node[main] (12) [below of=14] {$12$};
\node[main] (11) [below of=13] {$11$};

\node[main] (2) [ left of=4] {$2$};
\node[main] (6) [below of=2] {$6$};
\node[main] (5) [below of=4] {$5$};

\draw (8) -- (9);
\draw (8) -- (1);
\draw (8) -- (10);
\draw (8) -- (3);
\draw (8) -- (7);
\draw (9) -- (3);
\draw (9) -- (7);
\draw (1) -- (2);
\draw (1) -- (3);
\draw (1) -- (4);
\draw (1) -- (7);
\draw (10) -- (3);
\draw (10) -- (7);
\draw (10) -- (13);
\draw (10) -- (11);
\draw (10)[out=330, in=110] to (12);
\draw (10) -- (14);
\draw (3) -- (7);
\draw (3) -- (5);
\draw (3) -- (6);
\draw (5) -- (6);

\draw (7) -- (11);
\draw (7) -- (12);
\draw (11) -- (12);
\draw (3) -- (4);
\draw (5) -- (7);
\draw (1) -- (5);
\draw (1)[out=210, in=70] to (6);
\draw (7) -- (13);
\draw (3) -- (11);

\draw[red, dashed, very thick] (11) -- (13);
\draw[red, dashed, very thick] (5) -- (4);

\end{tikzpicture} 
   
}
\end{subfigure}\quad

\medskip
\centering
\begin{subfigure}[b]{0.5\textwidth}
\resizebox{6cm}{3cm}{%

\begin{tikzpicture}[node distance={15mm}, main/.style = {draw, circle}] 

\node[main] (8) {$8$}; 

\node[main] (9) [below of=8] {$9$}; 
\node[main] (1) [left of=9] {$1$};
\node[main] (10) [ right of=9] {$10$};
\node[main] (4) [below left of=1] {$4$}; 
\node[main] (13) [below right of=10] {$13$};
\node[main] (3) [below of=1] {$3$}; 
\node[main] (7) [below of=10] {$7$};
\node[main] (14) [ right of=13] {$14$};
\node[main] (12) [below of=14] {$12$};
\node[main] (11) [below of=13] {$11$};

\node[main] (2) [ left of=4] {$2$};
\node[main] (6) [below of=2] {$6$};
\node[main] (5) [below of=4] {$5$};

\draw (8) -- (9);
\draw (8) -- (1);
\draw (8) -- (10);
\draw (8) -- (3);
\draw (8) -- (7);
\draw (9) -- (3);
\draw (9) -- (7);
\draw (1) -- (2);
\draw (1) -- (3);
\draw (1) -- (4);
\draw (1) -- (7);
\draw (10) -- (3);
\draw (10) -- (7);
\draw (10) -- (13);
\draw (10) -- (11);
\draw (10)[out=330, in=110] to (12);
\draw (10) -- (14);
\draw (3) -- (7);
\draw (3) -- (5);
\draw (3) -- (6);
\draw (5) -- (6);

\draw (7) -- (11);
\draw (7) -- (12);
\draw (11) -- (12);
\draw (3) -- (4);
\draw (5) -- (7);
\draw (1) -- (5);
\draw (1)[out=210, in=70] to (6);
\draw (7) -- (13);
\draw (3) -- (11);

\draw[red, dashed, very thick] (5) -- (4);

\end{tikzpicture} 
   
}
\end{subfigure}

\caption{The five other graphs on 14 vertices which are interval and $K_{1,5}$-free but not disjoint unit 2-interval. In dashed red, the edges that differ from the graph that we analyze here.}
\label{pics:5counterexamples}
\end{figure}
\end{toappendix}

\begin{figure}[htp]
    \centering
        \centering

\begin{tikzpicture}[scale=0.65,transform shape,node distance={15mm}, main/.style = {draw, circle}] 

\node[main] (8) {$8$}; 

\node[main] (9) [below of=8] {$9$}; 
\node[main] (1) [left of=9] {$1$};
\node[main] (10) [ right of=9] {$10$};
\node[main] (4) [below left of=1] {$4$}; 
\node[main] (13) [below right of=10] {$13$};
\node[main] (3) [below of=1] {$3$}; 
\node[main] (7) [below of=10] {$7$};
\node[main] (14) [ right of=13] {$14$};
\node[main] (12) [below of=14] {$12$};
\node[main] (11) [below of=13] {$11$};

\node[main] (2) [ left of=4] {$2$};
\node[main] (6) [below of=2] {$6$};
\node[main] (5) [below of=4] {$5$};

\draw (8) -- (9);
\draw (8) -- (1);
\draw (8) -- (10);
\draw (8) -- (3);
\draw (8) -- (7);
\draw (9) -- (3);
\draw (9) -- (7);
\draw (1) -- (2);
\draw (1) -- (3);
\draw (1) -- (4);
\draw (1) -- (7);
\draw (10) -- (3);
\draw (10) -- (7);
\draw (10) -- (13);
\draw (10) -- (11);
\draw (10)[out=330, in=110] to (12);
\draw (10) -- (14);
\draw (3) -- (7);
\draw (3) -- (5);
\draw (3) -- (6);
\draw (5) -- (6);

\draw (7) -- (11);
\draw (7) -- (12);
\draw (11) -- (12);
\draw (3) -- (4);
\draw (5) -- (7);
\draw (1) -- (5);
\draw (1)[out=210, in=70] to (6);
\draw (7) -- (13);
\draw (3) -- (11);

\end{tikzpicture} 
   
    \caption{One of the 6 graphs with 14 vertices (the one with the fewest edges) which is an interval graph (see \cref{fig:representation}) and $K_{1,5}$-free, but not disjoint unit 2-interval.}
    \label{fig:counterexample}
\end{figure}

\begin{figure}[htp]
    \centering
        \centering
        
    \begin{tikzpicture}[scale=0.42]

    \draw (-3,6) -- (3,6)  node[midway,above]{\small $8$};
    \draw (-10,5) -- (-2,5)  node[midway,above] {\small $1$};
    \draw (-1,5) -- (1,5) node[midway,above] {\small $9$};
    \draw (2,5) -- (10,5) node[midway,above] {\small $10$};
    \draw (-10,3) -- (-9,3) node[midway,above] {\small $2$};
    \draw (-8,3) -- (-7,3) node[midway,above] {\small $4$};
    \draw (-6,3) -- (-5,3) node[midway,above] {\small $6$};
    \draw (-6,4) -- (-4,4) node[midway,above] {\small $5$};
 \draw (9,3) -- (10,3) node[midway,above] {\small $14$};
    \draw (7,3) -- (8,3) node[midway,above] {\small $13$};
    \draw (5,3) -- (6,3) node[midway,above] {\small $12$};
    \draw (4,4) -- (6,4) node[midway,above] {\small $11$};

    \draw (-8,2) -- (4.5,2) node[midway,above] {\small $3$};
    \draw (-4.5,1) -- (8,1) node[midway,above] {\small $7$};
    

    \end{tikzpicture}
    \caption{An interval representation of the graph in \cref{fig:counterexample}.}
    \label{fig:representation}
\end{figure}

\cref{conj:roberts} follows directly from the next lemma.

\begin{lemmarep}\label{lemma:certificate}
    The graph $G$ in~\cref{fig:counterexample} is not a disjoint unit 2-interval graph.
\end{lemmarep}

\begin{proof}
To prove that $G$ is not a disjoint unit 2-interval graph, we need to show that there is no way to obtain a unit interval graph $G'$ from $G$ by splitting some or all of the vertices $v$ of $G$ into two non-adjacent vertices $v_1$ and $v_2$, called the \emph{representatives} of $v$, 
so that $(u,v) \in E(G)$ if and only if $(u_i,v_j) \in E(G')$ for some $i,j \in \{1,2\}$. The edges $(u_i,v_j)$ in $E(G')$ are then called the \emph{representatives} of edge $(u,v) \in E(G)$, and $G'$ is called a \emph{solution}. Furthermore, we say that a solution is \emph{canonical} if for every vertex~$v$ that has been split into $v_1$ and $v_2$, none of the following hold:
\begin{itemize}
    \item $v_i$ is isolated in $G'$ (if this happened, we could remove $v_i$ from $G'$ and end up with another solution with fewer split vertices).
    \item $v_i$ is adjacent to $u_1$ or $u_2$ for every $(u,v) \in E(G)$ (in this case, we could remove vertex~$v_{3-i}$ from $G'$ and end up with another solution with fewer split vertices).
\end{itemize}
Note that if there exists a solution, then a canonical solution also exists. We now list some properties that a solution $G'$ must satisfy:

\begin{observation}
\begin{itemize}
    \item No edge $(u,v) \in E(G)$ can have 4 representatives in $G'$ since otherwise the subgraph induced by $(u_1,v_1,u_2,v_2)$ would be a $C_4$ in $G'$, contradicting the fact that $G'$ is an interval graph. 
    \item Every edge $(u,v)\in E(G)$ that is part of a 4-claw of $G$ must have precisely one representative in $G'$, since $G'$ has no 3-claw. 
    \item Every edge $(u,v)\in E(G)$ that is part of a 3-claw of $G$ has at most two representatives.
\end{itemize}
\end{observation}

After these preliminary observations, we proceed to the detailed analysis of the graph $G$ presented in~\cref{fig:counterexample} and we show that it does not have a canonical solution. We begin with three observations that follow directly from the structure of $G$.

\begin{observation}
\begin{itemize}
    \item No leaf vertex (vertex of degree 1) is split by a canonical solution. Hence, vertices~2 and~14 are not split in a canonical solution.
    \item The centers of an induced 3-claw in $G$ must be necessarily split in every solution. Therefore, if there exists a solution for $G$, then the vertices $1,3,7,8$ and $10$ must all be split.
    \item The graph $G$ has a unique non-trivial automorphism $\pi$, which is an involution and described by Table~1 (and it corresponds to the symmetry with respect to the y-axis in \cref{fig:counterexample}).    
    \end{itemize}
\end{observation}

    \begin{table}[h!]
    \mbox{\ } \hspace{1cm}
       \begin{tabularx}{0.85\textwidth}{|X||X|X|X|X|X|X|X|X|X|X|X|X|X|X|X|X|X|X|}
       \hline
          $i$ &  1 & 2 & 3 & 4 & 5 & 6 & 7 & 8 & 9 & 10 & 11 & 12 & 13 & 14 \\ \hline
          $\pi(i)$ & 10 & 14 & 7 & 13 & 11 & 12 & 3 & 8 & 9 & 1 & 5 & 6 & 4 & 2 \\
          \hline
       \end{tabularx} 
       \medskip
       \caption{The unique non-trivial automorphism of $G$.}
    \end{table}

Next, we observe that the 4-claws offer particularly strong constraints. Indeed, if a vertex~$v$ is the center of a 4-claw $[v;a,b,c,d]$ and we split it into two vertices $v_1$ and $v_2$, then there are precisely two vertices in ${a,b,c,d}$, say~$a$ and~$b$, such that $(v_1,a)$ and $(v_1,b)$ are the edges representing $(v,a)$ and $(v,b)$, whereas $(v_2,c)$ and $(v_2,d)$ are the edges representing $(v,c)$ and $(v,d)$.
In the following claims we examine and list some of the constraints on the possibilities when splitting the vertices of $G$.

\begin{claim}\label{node1}
Vertex~1 imposes the following constraints:
\begin{enumerate} 
    \item There exist three indices $i,j,k$ ($i,j,k\in\{1,2\}$) such that:
    \begin{itemize}
        \item \e{1_i}{7_j} is the unique representative of edge \e{1}{7},
        \item \e{1_i}{8_k} is the unique representative of edge \e{1}{8} and
        \item \e{7_j}{8_k} is the unique representative of edge \e{7}{8}.
    \end{itemize}
    
    \item There exist three indices $i,j,k$ ($i,j,k\in\{1,2\}$) such that:
    \begin{itemize}
        \item \e{1_i}{5_j} is the unique representative of edge \e{1}{5}
        \item \e{1_i}{6_k} is the unique representative of edge \e{1}{6} and
        \item \e{5_j}{6_k} is a representative of edge \e{5}{6}.
    \end{itemize}
\end{enumerate}
\end{claim}

\begin{claimproof}
Vertex~1 is the center of three different 4-claws: $A= [1;2,4,6,7], B= [1;2,4,6,8]$ and $ C= [1;2,4,5,8]$.
The first point of the claim follows from comparing claws $A$ and $B$. Towards a contradiction, assume without loss of generality that the representatives of $(1, 7)$ and $(1, 8)$ are incident to vertices $1_1$ and $1_2$, respectively. Then, by the pigeonhole principle, it is impossible to place the remaining representatives of the edges $(1,2)$, \e{1}{4} and \e{1}{6} without creating a 3-claw with center either $1_1$ or $1_2$ (we can have at most one of them incident to each of the vertices, and we have three edges and two vertices).
Furthermore, if they were both incident to $1_1$ (w.l.o.g.) but there was no representative of edge \e{7}{8} incident to $1_1$, then it would still be impossible to place the remaining edges without creating a claw. 
The proof of the second point of the claim follows similarly by comparing claws $B$ and $C$, but now the representative of \e{5}{6} is not necessarily unique because neither vertex is the center of a 4-claw.
\end{claimproof}

\begin{claim}\label{node10}
Vertex~10 imposes the symmetric constraints of vertex~1.
\end{claim}

\begin{claimproof}
        The proof is analogous to the previous one, comparing the symmetric 4-claws.
\end{claimproof}

\begin{claim}\label{node3}
Vertex~3 imposes the following constraints:
\begin{enumerate} 
    \item There exist three indices $i,j,k$ ($i,j,k\in\{1,2\}$) such that:
    \begin{itemize}
        \item \e{3_i}{5_j} is the unique representative of edge \e{3}{5},
        \item \e{3_i}{6_k} is the unique representative of edge \e{3}{6} and
        \item \e{5_j}{6_k} is a representative of edge \e{5}{6}.
    \end{itemize}
      
    \item There exist three indices $i,j,k$ ($i,j,k\in\{1,2\}$) such that:
    \begin{itemize}
        \item \e{3_i}{10_j} is the unique representative of edge \e{3}{10},
        \item \e{3_i}{11_k} is the unique representative of edge \e{3}{11}, and
        \item \e{10_j}{11_k} is the unique representative of edge \e{10}{11}.
    \end{itemize}
     \item There exist three indices $i,j,k$ ($i,j,k\in\{1,2\}$) such that:
    \begin{itemize}
        \item \e{3_i}{8_j} is the unique representative of edge \e{3}{8},
        \item \e{3_i}{9_k} is the unique representative of edge \e{3}{9}, and
        \item \e{8_j}{9_k} is a representative of edge \e{8}{9}.
    \end{itemize}
\end{enumerate}
\end{claim}

\begin{claimproof}
    Vertex~3 is the center of four different 4-claws:
    $A= [3;4,5,9,10], B= [3;4,6,9,10], C= [3;4,6,9,11]$ and $ D= [3;4,6,8,11]$.
    The first point follows by comparing claws $A$ and $B$; the second one, by comparing claws $B$ and $C$; and the third one, by comparing claws $C$ and $D$.    
\end{claimproof}

\begin{claim}\label{node7}
Vertex~7 imposes the symmetric constraints of vertex~3.
\end{claim}

\begin{claimproof}
        The proof is analogous to the proofs of the previous claims.       
\end{claimproof}

We now use the previous claims to deduce that the way to split vertex~8 is actually unique up to symmetry. 
First, it is clear that vertex~8 must be split because of the 3-claw $[8; 1,9,10]$. 
Assume vertex~8 is the first vertex of $G$ to be split. 

\begin{claim}
    There is a unique way (up to symmetry) to split vertex~$8$.
\end{claim}
\begin{claimproof}
Let $8_1$ be the representative of $8$ made adjacent to $10$. Then $8_1$ is also adjacent to 3 (by \cref{node10}, otherwise we get a 5-claw $[10; 8_1,3,14,13,12]$). Moreover, $8_1$ is also adjacent to $9$ (by \cref{node3}, otherwise we get a 5-claw $[3;8_1,9,4,6,11]$). 
As such, vertex~$1$ cannot be made adjacent to $8_1$ (because it would create a 3-claw) and has to be adjacent to $8_2$ instead. 
We have now the symmetric deductions.
Where $8_2$ is the representative of $8$ made adjacent to $1$, then $8_2$ is also adjacent to $7$ (by \cref{node1}, otherwise we get a 5-claw $[1; 8_2,7,2,4,6]$).
Moreover, $8_2$ is also adjacent to 9 (by \cref{node7}, otherwise we get a 5-claw $[7;8_2,9,13,12,5]$). Note that vertex~10 cannot be made adjacent to $8_2$. 
This shows that the way of splitting vertex~$8$ is fully determined 
before splitting the other vertices: one representative needs to be adjacent to $10, 3$ and $9$, while the other needs to see the vertices $1, 7$ and $9$.
The graph resulting from splitting vertex~8 can be seen in \cref{fig:split8}. 
\end{claimproof}

\begin{figure}
    \centering
        \centering

\begin{tikzpicture}[node distance={15mm}, main/.style = {draw, circle}] 

\node[main] (8) {$8_2$}; 

\node[main] (1) [below of=8] {$1$};
\node[main] (9) [right of= 1] {$9$}; 

\node[main] (10) [ right of=9] {$10$};
\node[main] (15)[ above of=10] {$8_1$}; 
\node[main] (4) [below left of=1] {$4$};

\node[main] (13) [below right of=10] {$13$};
\node[main] (3) [below of=1] {$3$}; 
\node[main] (7) [below of=10] {$7$};
\node[main] (14) [ right of=13] {$14$};
\node[main] (12) [below of=14] {$12$};
\node[main] (11) [below of=13] {$11$};

\node[main] (2) [ left of=4] {$2$};
\node[main] (6) [below of=2] {$6$};
\node[main] (5) [below of=4] {$5$};

\draw (8) -- (9);
\draw (8) -- (1);
\draw (15) -- (10);
\draw (15)[out=200, in=70] to (3);
\draw (15) -- (9);
\draw (8) [out=340, in=110] to (7);

\draw (9) -- (3);
\draw (9) -- (7);
\draw (1) -- (2);
\draw (1) -- (3);
\draw (1) -- (4);
\draw (1) -- (7);
\draw (10) -- (3);
\draw (10) -- (7);
\draw (10) -- (13);
\draw (10) -- (11);
\draw (10)[out=330, in=110] to (12);
\draw (10) -- (14);
\draw (3) -- (7);
\draw (3) -- (5);
\draw (3) -- (6);
\draw (5) -- (6);

\draw (7) -- (11);
\draw (7) -- (12);
\draw (11) -- (12);
\draw (3) -- (4);
\draw (5) -- (7);
\draw (1) -- (5);
\draw (1)[out=210, in=70] to (6);
\draw (7) -- (13);
\draw (3) -- (11);

\end{tikzpicture} 
   
    \caption{The graph $G$ after the split of vertex~8}
    \label{fig:split8}
\end{figure}

After the above split of vertex~8, we have one more 4-claw centered at vertex~3 (and, symmetrically, at vertex~7). These 4-claws introduce the following new constraints.

\begin{claim}\label{vertex3} 
Vertex~3 imposes now the following additional constraints:
\begin{enumerate}
    \item There exist three indices $i,j,k$ ($i,j,k\in\{1,2\}$) such that:
    \begin{itemize}
        \item \e{3_i}{7_j} is the unique representative of edge \e{3}{7},
        \item \e{3_i}{11_k} is the unique representative of edge \e{3}{6} and
        \item \e{7_j}{11_k} is the unique representative of edge \e{7}{11}.
    \end{itemize}
\end{enumerate}
Adding this to the previous constraints, we have that the unique representatives of the edges \e{7}{11} and \e{10}{11} should both end up within the neighborhood of either $3_1$ or $3_2$.
\end{claim}

\begin{claimproof}
This follows from comparing the 4-claws $D= [3;4,6,8,11]$ and $ E= [3;8_1,7,4,6]$.
\end{claimproof}

At this step, the constraints on vertex~3 allow us to derive the following result.

\begin{claim}
    The way to split vertex~3 is unique up to symmetry.
\end{claim}

\begin{claimproof}    
By the above, we already know:
\begin{enumerate}
    \item The representatives of \e{3}{7} and \e{3}{11} are incident to each other (because of the 4-claws $[3;4,6,8_1,7]$ and $[3;4,6,8_1,11]$).
    \item The representatives of \e{3}{11} and \e{3}{10} are incident to each other (because of the 4-claws $[3;4,6,9,11]$ and $[3;4,6,9,10]$).
    \item The representatives of \e{10}{3} and \e{10}{8_1} are incident to each other (because of the 4-claws $[10;14,13,12,3]$ and $[10;14,13,12,8_1]$) and to a representative of \e{3}{8}.
    \item The representatives of \e{3}{8_1} and \e{3}{9} are incident to each other (because of the 4-claws $[3;4,6,11,8_1]$ and $[3;4,6,11,9]$).
    \item There is an unique pair $i,j$ such that $3_i$ and $10_j$ are adjacent (since the edge belongs to a 4-claw $[10;14,13,12,3]$).
\end{enumerate}

Points 1 and 2 imply that one of the vertices representing 3, w.l.o.g., $3_1$, is adjacent to vertices 7, 11 and 10. Let $10_1$ be the vertex representing 10 and adjacent to $3_1$ (well defined by point 5). Then $10_1$ is adjacent to $8_1$ and it must also be the case that $3_1$ and $8_1$ are adjacent. Otherwise, we would obtain the 5-claw $[10;14,12,13,3_1,8_1]$ in the original certificate graph. By point 4, since vertex~$3_1$ is adjacent to $8_1$, it must also be adjacent to 9.
Now, $3_1$ cannot have any more neighbors, as if we add any new edges, we would create a $K_{1,3}$. Therefore, the way of splitting vertex~$3$ is uniquely determined up to symmetry: one of the representatives of $3$ is adjacent to 7, 11, 10, 8 and 9; and the other one, to the rest of its neighbors.
The graph resulting from splitting vertex~3 can be seen in \cref{fig:split3}.
\end{claimproof}

\begin{figure}
    \centering
        \centering

\begin{tikzpicture}[scale=0.9,transform shape,node distance={15mm}, main/.style = {draw, circle}] 

\node[main] (8) {$8_2$}; 

\node[main] (1) [below of=8] {$1$};
\node[main] (9) [right of= 1] {$9$}; 

\node[main] (10) [ right of=9] {$10$};
\node[main] (15)[ above of=10] {$8_1$}; 


\node[main] (1) [left of=9] {$1$};
\node[main] (4) [below left of=1] {$4$}; 
\node[main] (13) [below right of=10] {$13$};
\node[main] (3) [below of=1] {$3_2$}; 
\node[main] (16) [right of=3] {$3_1$}; 

\node[main] (7) [below of=10] {$7$};
\node[main] (14) [ right of=13] {$14$};
\node[main] (12) [below of=14] {$12$};
\node[main] (11) [below of=13] {$11$};

\node[main] (2) [ left of=4] {$2$};
\node[main] (6) [below of=2] {$6$};
\node[main] (5) [below of=4] {$5$};

\draw (8) -- (9);
\draw (8) -- (1);
\draw (15) -- (10);
\draw (15) -- (9);
\draw (8) [out=340, in=110] to (7);

\draw (9) -- (7);
\draw (1) -- (2);

\draw (1) -- (4);
\draw (1) -- (7);
\draw (10) -- (16);
\draw (10) -- (7);
\draw (10) -- (13);
\draw (10) -- (11);
\draw (10)[out=330, in=110] to (12);
\draw (10) -- (14);

\draw (9) -- (16);
\draw (16) -- (7);
\draw (16) -- (11);
\draw (3) -- (5);
\draw (3) -- (6);
\draw (3) -- (4);

\draw (1) -- (3);
\draw (15) -- (16);

\draw (5) -- (6);

\draw (7) -- (11);
\draw (7) -- (12);
\draw (11) -- (12);

\draw (5) -- (7);
\draw (1) -- (5);
\draw (1)[out=210, in=70] to (6);
\draw (7) -- (13);

\end{tikzpicture} 
   
    \caption{Graph resulting from the split of vertex~3}
    \label{fig:split3}
\end{figure}

Finally, we consider vertex~7 and we state that it cannot be split without creating a $K_{1,3}$.

\begin{claim}\label{split7}
    There is no way to split vertex~7 without creating a $K_{1,3}$ with center either $7_1$ or $7_2$.
\end{claim}
\begin{claimproof}
If we apply the symmetry of the graph before splitting vertex~3, we would obtain that one of the representatives of vertex~$7$, say $7_1$ has to be adjacent to 1, 8, 9, 3 and 5. 
However, since vertices 3 and 8 have already been split, these five edges now create the 3-claw $[7; 8_2, 3_1, 5]$. 
Furthermore, note that $3_1$ cannot be made adjacent to the other representative of vertex~7, as it also creates a 3-claw ($[7; 3_1, 12,13]$).
Thus, it is not possible to split vertex~7 maintaining $G'$ as a unit interval graph. 
\end{claimproof}

The last claim concludes the proof that the $K_{1,5}$-free interval graph presented is not disjoint unit 2-interval. 
\end{proof}


We conclude this section showing that \cref{conj:roberts} can actually be generalized for disjoint unit $d$-interval graphs for any $d > 2$.

\begin{corollaryrep}\label{corollaryd}
There exists a $K_{1, 2d+1}$-free interval graph that is not a disjoint unit $d$-interval graph.
\end{corollaryrep}
\begin{proof}
    One can extend the example for $d=2$ (displayed in \cref{fig:counterexample}) by adding $2d-4$ common neighbors to vertices 1 and 3 and $2d-4$ common neighbors to vertices 7 and 10. The proof that this graph is not disjoint unit $d$-interval is analogous to the proof for $d=2$.
\end{proof}

\section{Inclusions between the different subclasses of d-interval graphs}\label{SectionBalanced}

In this section, we analyze the relationships between different subclasses of multiple interval graphs. We have already seen that 2-interval graphs and disjoint 2-interval graphs are equivalent. 
Furthermore, the results from the previous two sections imply that the class of disjoint unit 2-interval graphs is properly contained in the class of unit 2-interval graphs. 
In the following, we summarize the containment relationships between unit 2-interval graphs, disjoint unit 2-interval graphs, balanced 2-interval graphs and disjoint balanced 2-interval graphs (see
\cref{fig:classes free} for a graphical illustration).

\begin{figure}[!ht]
    \centering
    \begin{tikzpicture}[scale=0.8, transform shape]
    \usetikzlibrary{arrows}
    \usetikzlibrary{shapes}
    \tikzstyle{every node}=[draw=black, minimum width=0pt, align=center,rounded corners]
    \tikzset{>=Stealth}

    \node[text width=5cm] (f) at (0, 3) {disjoint 3-interval $=$ 3-interval};
    \node[text width=5cm] (g) at (0, 2) { balanced 3-interval};
    \node[fill=white,text width=5cm] (h) at (0, 1) {disjoint balanced $3$-interval};
        
    \node[text width=5cm] (i) at (-7, 1.5) {disjoint 2-interval $=$ 2-interval};
    
    \node[text width=5cm] (c) at (0, 0) {disjoint balanced $2$-interval $=$ balanced 2-interval};
    \node[text width=5cm] (b) at (0,-1) {unit 2-interval};

    \node[text width=5cm] (a) at (0, -2) {disjoint unit 2-interval};


    
    \draw[->] (c) -- (b);
    
    \draw[->] (h) -- (c);
    
    \draw[->] (f) -- (g);
    \draw[->] (g) -- (h);
    \draw[->,rounded corners=5pt] (f) -- (-7,3) -- (i);
    \draw[->,rounded corners=5pt] (i) -- (-7,0) --  (c);
    \draw[->] (b) -- (a);


    \end{tikzpicture}
    \caption{\label{fig:classes free}
    Landscape of graph subclasses of 3-interval graphs. 
    An arrow from a graph class $\mathcal{C}$ to a $\mathcal{C'}$ indicates that $\mathcal{C'} \subset \mathcal{C}$. The relationships between the class of 2-interval graphs and the classes of balanced 3-interval graphs and disjoint balanced 3-interval graphs are not known. 
    }
\end{figure}



\begin{theoremrep}
\label{thm:summary}
\begin{enumerate}
    \item The classes of \emph{2-interval} and \emph{disjoint 2-interval} graphs are \emph{equivalent}.
    \item The classes of \emph{balanced 2-interval} and \emph{disjoint balanced 2-interval} graphs are \emph{equivalent}.
    \item The class of \emph{unit 2-interval graphs} is \emph{properly contained} in the class of \emph{disjoint balanced 2-interval} graphs. 
    \item The class of \emph{disjoint unit 2-interval} graphs is \emph{properly contained} in the class of \emph{unit 2-interval} graphs.  
\end{enumerate}
    
\end{theoremrep}

\begin{proof}
    \begin{enumerate}
    
        \item See \cref{obs:disjoint free equivalent}.
        \item \label{balancedequiv} To see that the classes of balanced and disjoint balanced 2-intervals are equivalent, it suffices to observe that if we have a balanced 2-interval representation, then for every pair of intersecting intervals $[a,b]$ and $[c,d]$ (with $a\leq c\leq b\leq d$) associated to the same vertex, we can stretch both intervals until the middle point of their intersection, i.e., we can replace both intervals by $[a, a+\frac{b-c}{2}-\epsilon]$ and $ [a+ \frac{b-c}{2}+\epsilon, d]$, respectively, for some $\epsilon$ sufficiently small. This procedure yields a disjoint balanced 2-interval representation.
        \item To see that unit 2-interval graphs have a disjoint balanced 2-interval representation, it suffices to apply the same procedure as in \ref{balancedequiv}. Now the two obtained intervals won't have unit length any longer, but the length of both intervals will be the same.
    The inclusion is proper because $K_{5,3}$ is disjoint balanced 2-interval (see \cite[Fig. 3]{DBLP:conf/wg/GambetteV07}) but not unit 2-interval (the latter follows from the fact that every vertex of degree 3 is adjacent to 5 independent vertices, so one of the intervals must intersect more than two disjoint intervals).
    \item As shown before, the graph in \cref{fig:counterexample} is not disjoint unit 2-interval, but applying the algorithm presented in Section~\ref{SectionAlgorithm}, one can obtain a unit 2-interval representation (see \cref{fig:2overlap}) of it. 
This proves that the class of disjoint unit 2-interval graphs is properly contained in the class of unit 2-interval graphs.
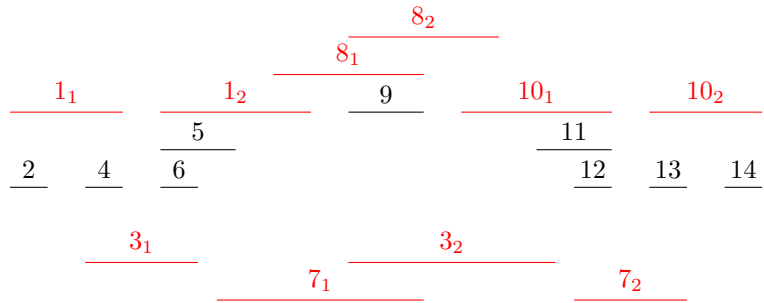
\begin{figure}[h]
    \centering
        \centering
        
    \begin{tikzpicture}[scale=0.5]

    \draw[red] (-3,6) -- (1,6)  node[midway,above]{$8_1$};
        \draw[red] (-1,7) -- (3,7)  node[midway,above]{$8_2$};

    \draw[red] (-10,5) -- (-7,5)  node[midway,above] {$1_1$};
     \draw[red] (-6,5) -- (-2,5)  node[midway,above] {$1_2$};
    \draw (-1,5) -- (1,5) node[midway,above] {$9$};
    \draw[red] (2,5) -- (6,5) node[midway,above] {$10_1$};
        \draw[red] (7,5) -- (10,5) node[midway,above] {$10_2$};

    \draw (-10,3) -- (-9,3) node[midway,above] {$2$};
    \draw (-8,3) -- (-7,3) node[midway,above] {$4$};
    \draw (-6,3) -- (-5,3) node[midway,above] {$6$};
    \draw (-6,4) -- (-4,4) node[midway,above] {$5$};
 \draw (9,3) -- (10,3) node[midway,above] {$14$};
    \draw (7,3) -- (8,3) node[midway,above] {$13$};
    \draw (5,3) -- (6,3) node[midway,above] {$12$};
    \draw (4,4) -- (6,4) node[midway,above] {$11$};

    \draw[red] (-8,1) -- (-5,1) node[midway,above] {$3_1$};
        \draw[red] (-1,1) -- (4.5,1) node[midway,above] {$3_2$};

    \draw[red] (-4.5,0) -- (1,0) node[midway,above] {$7_1$};
        \draw[red] (5,0) -- (8,0) node[midway,above] {$7_2$};


    \end{tikzpicture}
    \caption{Unit 2-interval representation of the graph in \cref{fig:counterexample}.}
    \label{fig:2overlap}
\end{figure}
    \end{enumerate}
 
\end{proof}

We finish by showing that the previous theorem cannot be completely generalized for the subclasses of $d$-interval graphs, as the class of balanced $d$-interval graphs is not equivalent to the class of disjoint balanced $d$-interval graphs for $d>2$. We first construct a graph that is balanced 3-interval but not disjoint balanced 3-interval and then show how to generalize this construction for every $d>3$.

    \begin{theorem}\label{disjointbalanced}
The class of disjoint balanced 3-interval graphs is properly contained in the class of balanced 3-interval graphs.
\end{theorem}

\begin{proof}
    We construct a graph $G$ which is balanced 3-interval but not disjoint balanced 3-interval.  
    The high-level idea of the construction is that for a particular vertex, one of its intervals is forced to a given length, while the other two are forced to be placed somewhere where there is not enough space for both of them, and thus they cannot be disjoint (note that the difference with the case $d=2$ is that now, if we stretch two of the intervals so that they do not intersect, we also have to modify the length of the third interval, and as we show here, this is not always possible). To enforce these constraints, we use the complete bipartite graph $K_{11,4}$ as a gadget and exploit the fact that any 3-interval representation of this gadget must be continuous (i.e., the union of the intervals in its underlying family is an interval)~\cite[Lemma 2]{DBLP:journals/dam/WestS84} (see also \cite[Fig. 3]{DBLP:conf/wg/GambetteV07} for the idea of its representation). 
    
    We construct $G$ as follows: we connect in a chain five $K_{11,4}$'s, to which we add six vertices $v_1$, $v_2$, $v_3$, $v_4$, $v_5$, $v_6$ (\cref{fig:bal3} shows how to link $v_1$, $v_2$, $v_3$, $v_4$ to the chain, while vertices $v_5$ and $v_6$ mimic the behavior of $v_3$ and $v_4$ with a different set of neighbors, namely, $v_5$ is connected to the corresponding vertices of the first two $K_{11,4}$'s, and $v_6$ is connected to the corresponding vertices of the second and the third $K_{11,4}$'s).
More precisely, let $C_i$, with $i\in\{1, \dots, 5\}$, be the five $K_{11,4}$'s forming the chain, enumerated from left to right. Moreover, for every $C_i$, let $f_i^j$ with $j\in \{1,\dots, 11\}$ be the eleven vertices of one side of the bipartition, and $t_i^k$, with $k\in \{1,2,3,4\}$, the four vertices of the other side of the bipartition. We assume that the chain is connected such that $f_i^{11}$ is linked to $f_{i+1}^1$. Then, $v_1$ is connected to all the vertices of $C_2$ and $C_4$, and to $f_3^{11}$ and $f_5^1$, plus another independent vertex.
Similarly, $v_2$ is connected to all the vertices of $C_2$ and $C_4$, to $v_1$ and to $f_1^{11}$ and $f_3^1$. 
On the other hand, $v_3$ is connected to $f_3^{11}$, 
$t_3^4$, $t_4^1$ and $f_4^j$ for $j\in \{1,\dots, 9\}$; 
while $v_4$ is connected to $f_5^1$, $t_4^4$, $t_5^1$ and $f_4^j$ for $j\in \{3,\dots, 11\}$. 
Finally, $v_5$ is connected to $f_1^{11}$, $t_1^3$, $t_2^1$ and $f_2^j$ for $j\in \{1,\dots, 7\}$, as well as $f_4^8$ and $f_4^9$, whereas $v_6$ is connected to $f_3^1$, $t_2^4$, $t_3^1$ and $f_2^j$ for $j\in \{5,\dots, 11\}$, as well as $f_4^8$ and $f_4^9$. The vertices $v_1$ and $v_2$ are both connected to $v_3, v_4, v_5$ and $v_6$.
    
    Now, as any 3-interval representation of a $K_{11,4}$ is continuous, any realization of $G$ groups the five $K_{11,4}$'s in a block~\cite{DBLP:journals/dam/WestS84}. 
    For $j\in \{1,2,3\}$, let $I_j$ be the intervals associated to $v_1$, $J_j$ the intervals associated to $v_2$, and $K_j$ the intervals associated to $v_3$. 
    First, it is clear that we need three different intervals to cover the neighbors of $v_1$ (and these three intervals must be disjoint). 
    Instead, the neighbors of $v_2$ could be covered only with two intervals.
    However, we will see that the two segments of the real line that need to be covered cannot have the same length (assuming that the 3-interval associated to $v_1$ is balanced).
    We will show that we need two intersecting intervals to cover the first segment.

    Suppose that only two intervals are needed to represent the adjacencies of $v_2$, and let $J$ be the interval displaying the edges between $C_2$ and $v_2$, and $J_3$ the interval displaying the edges between $C_4$ and $v_2$. Similarly, let $I_1$ be the interval associated to $v_1$ used to represent the edges with $C_2$, let $I_2$ the interval used to represent the edges with $C_4$, and $I_3$ the interval displaying the edge with the isolated vertex.
    One can easily see that $J_3$ is properly contained in $I_2$ (since $I_2$ must also intersect an interval associated to $f_3^{11}$ on its left and an interval associated to $f_5^1$ on its right), while $I_1$ is properly contained in $J$ (by an analogous argument). 
    Thus, $len(J_3) < len(I_2) =len(I_1) < len(J)$. 
    In order for the representation to be balanced, the segment of the real line covered by $J$ needs to be covered by two different intervals, say $J_1$ and $J_2$.
    To prove that $G$ is balanced 3-interval but not disjoint balanced 3-interval, we need to bound $len(J) - len(J_3)$. In particular, we need $len(J) - len(J_3) < len(J_3)$. Vertices $v_3$ and $v_4$ will allow us to find constants $a$ and $a'$ to bound $len(I_2) -len(J_3) \leq a+a'$, while vertices $v_5$ and $v_6$ will serve to find constants $b$ and $b'$ to bound $len(J)-len(I_1) \leq b+b' $. By showing that we can force the constants such that $a + a' + b + b' < len(J_3)$, we have the result. This will follow since we will have eight pairwise disjoint intervals properly contained in $J_3$: two of length $a$, two of length $a'$, two of length $b$ and two of length $b'$.
    
    Indeed, let $a$ and $a'$ be the lengths of the intervals associated to $v_3$ and to $v_4$, respectively. The next claim implies that there are two disjoint intervals associated to $v_3$ properly contained in $J_3$, and another disjoint interval that properly contains the segment between $l(I_2)$ and $l(J_3)$, and so $l(J_3)-l(I_2)< a$. 

\begin{claimrep}{\label{v3disjoint}}
Let $G$ be a graph formed by the union of a $K_{11,4}$ and a vertex $v$ which is adjacent to nine vertices in $S_{11}$, where $S_{11}$ denotes the side of the bipartition with eleven vertices. Then, vertex $v$ must be represented by three pairwise disjoint intervals, two of which are each properly contained in an interval representing a vertex of $S_{11}$.
\end{claimrep}
\begin{claimproof}
    Let $S_{11}$ and $S_4$ denote the two sides of the bipartition of the $K_{11,4}$, and let $f^j$ with $j \in \{1,\ldots, 11\}$ denote the vertices of $S_{11}$ and $t^j$ with $j \in \{1,\ldots, 4\}$ denote the vertices of $S_4$.
    Since every interval of $S_{11}$ intersects four pairwise disjoint intervals, every $f^j$ must be represented by three intervals: one that intersects two intervals $t^i, t^j$, with $i,j \in \{1,\ldots,4\}$, and two intervals contained each in one of the remaining $t^i$'s. 
    Thus, for any given $t^i$, there are at least four vertices $f^j$ such that no interval associated to them is contained in an interval of $t^i$. Since there are eleven vertices in $S_{11}$ and nine are adjacent to $v$, at least two of these nine vertices will not be contained in $t^1$. In particular, they will also not intersect the first interval of $t^1$. Thus, one of the intervals of $v$ must intersect $t^1$ and seven of the $f_4^j$
    , while the other two intervals of $v$ are contained each in one of the intervals of the remaining $f^j$'s, w.l.o.g., $f^8$ and $f^9$ (recall that the two remaining $f^j$'s fill the gaps in between intervals of the from $t^i$, and so it is possible to display these intersections without creating any new ones). 
    
    Finally, the reader can observe that this gadget admits a balanced 3-interval representation, as one can dilate the two holes that are covered by $f^8$ and $f^9$ as needed (and since no interval of these vertices is contained in an interval of $t^1$, modifying their length will have no effect on the length of the first interval of $v$).
    More precisely, consider a representation where an interval of $ t^1 $ contains intervals of $ f^1, f^2, f^3, f^4, f^5, f^6 $ and intersects $ f^7 $; an interval of $ t^4 $ contains $ f^3, f^4, f^5, f^7, f^{10}, f^{11} $ and intersects $ f^6 $; and an interval of $ v $ contains $ t_1 $. Then, $ t_2 $ and $ t_3 $ must contain an interval of $ f^8 $ and $ f^9 $, which must have length at least the length of $ v $ plus one (because another interval associated with them contains an interval of $ v $). On the other hand, $ t_2 $ also contains $ f^1, f^2 $, and $ f^6 $, while $ t_3 $ contains $ f^7, f^{10} $, and $ f^{11} $. Thus, let all $ f^i $ have length three except $ f^8 $ and $ f^9 $. We can then take $ t^1 $ and $ t^2 $ to have length $ 6(3+1)+1=25 $. Therefore, $ l(v) \geq 25 $, and it suffices to take $ l(f^8)=l(f^9)= l(v)+ 4 $. Finally, we can take $ l(t_2)=l(t_3) = 2(l(v)+5) +3(4) + 3 $. In fact, note that we can take $l(f^8)$ and $l(f^9)$ arbitrarily large, and so, in the complete construction, an interval associated to $f_4^8$ and an interval associated to $f_4^9$ can each contain intervals associated to $v_4,v_5$ and $v_6$, while still maintaining the balanced property. 
\end{claimproof}
  
    
    Similarly, the segment between $r(I_2)$ and $r(J_3)$ is also contained in an interval associated to $v_4$, which has the same properties as $v_3$ 
    and does not intersect any interval associated to $v_3$. This proves that there are two intervals of length $a$ and two intervals of length $a'$ (all pairwise disjoint) contained in $J_3$.  
    Doing the same to bound $l(I_1) -l(J)$ and $r(J)-l(I_1)$, we get the result.
    Thus, to represent $v_2$, we need two intervals associated to $v_2$ to intersect. If we do not allow intersection, the length of these two intervals will be smaller than the length of the third interval associated to $v_2$, contradicting the fact that they are balanced.

\begin{figure}
    \centering
    \usetikzlibrary{shapes.geometric}

\begin{tikzpicture}[scale=0.8]
  \foreach \x in {1,2,3,4,5} {
    \node[draw,ellipse,minimum width=1.5cm,minimum height=0.5cm] (A\x) at (\x*2,1) {};
    \foreach \y in {1,2,3,4,5,6,7,8,9,10,11} {
     \filldraw[black] (\x*2-0.7+0.11*\y,1) circle (0.5pt);
    \node (p\x\y) at (\x*2-0.7+0.11*\y,1) {};

    }
    \node[draw,ellipse,minimum width=1cm,minimum height=0.4cm] (B\x) at (\x*2,0) {};
    \foreach \y in {1,2,3,4} {
      \filldraw[black] (B\x) +(-0.4+0.15*\y,0) circle (1pt);
        \node (q\x\y) at (\x*2-0.4+0.15*\y,0) {};

    }
    
    \draw (A\x) -- (B\x);
  
 \node[circle,fill=blue,label=$v_1$, minimum size=1mm] (v1) at (8,3) {};
 \node[circle,fill=green,label=$v_2$] (v2) at (5,3) {};
 \node[circle,fill=red,label=below:{$v_3$}] (v3) at (6.5,-2) {};
 \node[circle,fill=violet,label=below:{$v_4$}] (v4) at (9.5,-2) {};
 
\filldraw[blue] (10,2) circle (2pt);

  
  }

  \draw (p111.center) -- (p21.center);
  \draw (p211.center) -- (p31.center);
  \draw (p311.center) -- (p41.center);
  \draw (p411.center) -- (p51.center);
  
  \draw [line width=0.4mm, color=green] (A2) -- (v2);
    \draw [line width=0.4mm, color=green] (B2) edge[bend right=30] (v2);
 \draw [line width=0.4mm, color = green] (A4) -- (v2);
    \draw [line width=0.4mm, color = green] (B4)  -- (v2);
    \draw [line width=0.01mm, color=green] (p111.center)  -- (v2);
    \draw [line width=0.01mm, color=green] (p31.center) -- (v2);

  \draw [line width=0.3mm, color = blue] (A4) -- (v1);
    \draw [line width=0.3mm, color = blue] (B4)  [bend right=10] to [bend right=45] (v1);
     \draw [line width=0.3mm, color=blue] (A2) -- (v1);
    \draw [line width=0.4mm, color=blue] (B2) edge[bend left=10] (v1);
    
    \draw [line width=0.01mm, color=blue] (p311.center) -- (v1);
    \draw [line width=0.01mm, color=blue] (p51.center) -- (v1);
    \draw [line width=0.01mm, color=blue] (10,2) -- (v1);
    \draw [line width=0.01mm, color=blue] (v1) -- (v2);

     \draw [line width=0.01mm, color = red] (p311.center) -- (v3);
    \draw [line width=0.01mm, color = red] (q41.center)  -- (v3);
     \draw [line width=0.01mm, color=red] (p41.center) -- (v3);
        \draw [line width=0.01mm, color = red] (p42.center) -- (v3);
     \draw [line width=0.01mm, color = red] (p43.center) -- (v3);

     \draw [line width=0.01mm, color = red] (p44.center) -- (v3);
     \draw [line width=0.01mm, color = red] (p45.center) -- (v3);
     \draw [line width=0.01mm, color = red] (p46.center) -- (v3);
     \draw [line width=0.01mm, color = red] (p47.center) -- (v3);
     \draw [line width=0.01mm, color = red] (p48.center) -- (v3);
     \draw [line width=0.01mm, color = red] (p49.center) -- (v3);

    \draw [line width=0.01mm, color=red] (3*2+0.2, 0) -- (v3);


     \draw [line width=0.01mm, color = violet] (q44.center) -- (v4);
     \draw [line width=0.01mm, color=violet] (p48.center) -- (v4);
    \draw [line width=0.01mm, color=violet] (p49.center) -- (v4);
    
      \draw [line width=0.01mm, color=violet] (q51.center) -- (v4);
      \draw [line width=0.01mm, color=violet] (p410.center) -- (v4);
      \draw [line width=0.01mm, color=violet] (p411.center) -- (v4);
      \draw [line width=0.01mm, color=violet] (p47.center) -- (v4);
      \draw [line width=0.01mm, color=violet] (p46.center) -- (v4);
      \draw [line width=0.01mm, color=violet] (p45.center) -- (v4);
      \draw [line width=0.01mm, color=violet] (p44.center) -- (v4);
      \draw [line width=0.01mm, color=violet] (p43.center) -- (v4);
      \draw [line width=0.01mm, color=violet] (p51.center) -- (v4);

      \draw[color=green] (v2) to[out=-180,in=90] (1, 1) to[out=-90,in=180] (v3);
      \draw[color=green] (v2) to[out=-180,in=90] (1, 1) to[out=-90,in=220] (v4);
      \draw[color=blue] (v1) to[out=0,in=90] (12, 0) to[out=-90,in=-70] (v3);
      \draw[color=blue] (v1) to[out=0,in=90] (12, 0) to[out=-90,in=0] (v4);
\end{tikzpicture}
    \caption{$G$ is balanced 3-interval but not disjoint balanced 3-interval. $K_{11,4}$ graphs are drawn abstractly and are chained. A thick edge stopping at the border of the ellipse means that the vertex is connected to every vertex in the corresponding part of the $K_{11,4}$. Vertices $v_5$ and $v_6$ are omitted for readability purposes.}
    \label{fig:bal3}
\end{figure}
\end{proof}

\begin{corollaryrep}
The class of disjoint balanced $d$-interval graphs is properly contained in the class of balanced $d$-interval graphs for every natural number $d \geq 3$.
\end{corollaryrep}
\begin{proof}
    The construction used in the proof of \cref{disjointbalanced} can be modified to yield a balanced $d$-interval graph by adding new isolated neighbors to $v_1$ and $v_2$ and replacing the chain of $K_{11,4}$'s by a chain of $K_{d^2+d-1, d+1}$'s, as any $d$-interval representation of this gadget must be continuous~\cite{DBLP:journals/dam/WestS84}. 
    Furthermore, we modify the adjacencies of vertices $v_3,v_4,v_5$ and $v_6$: vertex $v_3$ is now connected to $f_3^{11}$, 
$t_3^4$, $t_4^1$ and $f_4^j$ for $j\in \{1, \dots, d^2-d+1-2d+4\}$; 
while $v_4$ is connected to $f_5^1$, $t_4^4$, $t_5^1$ and $f_4^j$ for $j\in \{2d-3, \dots, d^2+d-1\}$, and the adjacencies of $v_5$ and $v_6$ are also modified accordingly. The same arguments used in \cref{disjointbalanced} prove that this constructed graph is not a disjoint balanced $d$-interval graph.
\end{proof}

\section{Concluding remarks}\label{SectionConclusion}
We have shown that the natural generalization of Roberts characterization for unit interval graphs remains valid for the most general definition of $d$-interval graphs that are interval graphs. 
However, quite surprisingly, if we require the $d$ intervals to be disjoint, then the result does not hold anymore. It remains as an open question whether disjoint $d$-interval graphs that are also interval can be characterized in some other way, or simply if they can be recognized in polynomial time. 
Finally, we have obtained a relatively complete landscape of the containment relationships between different subclasses of 2-interval graphs, that cannot be fully generalized for $d > 2$. In particular, for $d>2$, it is still unknown whether the class of unit $d$-interval graphs is contained in the class of disjoint balanced $d$-interval graphs. 


\bibliographystyle{plainurl}
\bibliography{bib}

\begin{thebibliography}{WS84}

\bibitem[GV07]{DBLP:conf/wg/GambetteV07}
Philippe Gambette and St{\'{e}}phane Vialette.
\newblock On restrictions of balanced 2-interval graphs.
\newblock In {\em {WG} 2007}, volume 4769 of {\em LNCS}, pages 55--65. Springer, 2007.

\bibitem[WS84]{DBLP:journals/dam/WestS84}
Douglas~B. West and David~B. Shmoys.
\newblock Recognizing graphs with fixed interval number is {NP}-complete.
\newblock {\em Discret. Appl. Math.}, 8(3):295--305, 1984.

\end{thebibliography}


\begin{thebibliography}{10}

\bibitem{acceptedvirginia}
Virginia Ardévol~Martínez, Romeo Rizzi, Florian Sikora, and Stéphane Vialette.
\newblock Recognizing unit multiple intervals is hard.
\newblock In Satoru Iwata and Naonori Kakimura, editors, {\em 34th International Symposium on Algorithms and Computation, {ISAAC} 2023}, volume 283 of {\em LIPIcs}, pages 8:1--8:18. Schloss Dagstuhl - Leibniz-Zentrum f{\"{u}}r Informatik, 2023.
\newblock URL: \url{https://doi.org/10.4230/LIPIcs.ISAAC.2023.8}, \href {https://doi.org/10.4230/LIPICS.ISAAC.2023.8} {\path{doi:10.4230/LIPICS.ISAAC.2023.8}}.

\bibitem{bar2006scheduling}
Reuven Bar-Yehuda, Magn{\'u}s~M Halld{\'o}rsson, Joseph Naor, Hadas Shachnai, and Irina Shapira.
\newblock Scheduling split intervals.
\newblock {\em {SIAM} J. Comput.}, 36(1):1--15, 2006.

\bibitem{DBLP:journals/dm/BogartW99}
Kenneth~P. Bogart and Douglas~B. West.
\newblock A short proof that 'proper = unit'.
\newblock {\em Discret. Math.}, 201(1-3):21--23, 1999.
\newblock \href {https://doi.org/10.1016/S0012-365X(98)00310-0} {\path{doi:10.1016/S0012-365X(98)00310-0}}.

\bibitem{DBLP:journals/jcss/BoothL76}
Kellogg~S. Booth and George~S. Lueker.
\newblock Testing for the consecutive ones property, interval graphs, and graph planarity using {PQ}-{T}ree algorithms.
\newblock {\em J. Comput. Syst. Sci.}, 13(3):335--379, 1976.
\newblock \href {https://doi.org/10.1016/S0022-0000(76)80045-1} {\path{doi:10.1016/S0022-0000(76)80045-1}}.

\bibitem{butman2010optimization}
Ayelet Butman, Danny Hermelin, Moshe Lewenstein, and Dror Rawitz.
\newblock Optimization problems in multiple-interval graphs.
\newblock {\em {ACM} Trans. Algorithms}, 6(2):1--18, 2010.

\bibitem{DBLP:journals/siamdm/CorneilOS09}
Derek~G. Corneil, Stephan Olariu, and Lorna Stewart.
\newblock The {LBFS} structure and recognition of interval graphs.
\newblock {\em {SIAM} J. Discret. Math.}, 23(4):1905--1953, 2009.
\newblock \href {https://doi.org/10.1137/S0895480100373455} {\path{doi:10.1137/S0895480100373455}}.

\bibitem{DBLP:journals/tcs/CrochemoreHLRV08}
Maxime Crochemore, Danny Hermelin, Gad~M. Landau, Dror Rawitz, and St{\'{e}}phane Vialette.
\newblock Approximating the 2-interval pattern problem.
\newblock {\em Theor. Comput. Sci.}, 395(2-3):283--297, 2008.
\newblock \href {https://doi.org/10.1016/j.tcs.2008.01.007} {\path{doi:10.1016/j.tcs.2008.01.007}}.

\bibitem{Duran2016}
Guillermo Dur{\'{a}}n, Florencia~Fern{\'{a}}ndez Slezak, Luciano~N. Grippo, Fabiano de~S. Oliveira, and Jayme~Luiz Szwarcfiter.
\newblock On unit d–interval graphs.
\newblock VII Latin American Workshop on Cliques in Graphs, 2016.
\newblock URL: \url{https://www.mate.unlp.edu.ar/~liliana/lawclique_2016/ffslezak.pdf, https://www.mate.unlp.edu.ar/~liliana/lawclique_2016/prolist.pdf}.

\bibitem{erdos1985note}
Paul Erd{\"o}s and Douglas~B West.
\newblock A note on the interval number of a graph.
\newblock {\em Discret. Math.}, 55(2):129--133, 1985.

\bibitem{DBLP:journals/tcs/FellowsHRV09}
Michael~R. Fellows, Danny Hermelin, Frances~A. Rosamond, and St{\'{e}}phane Vialette.
\newblock On the parameterized complexity of multiple-interval graph problems.
\newblock {\em Theor. Comput. Sci.}, 410(1):53--61, 2009.
\newblock \href {https://doi.org/10.1016/j.tcs.2008.09.065} {\path{doi:10.1016/j.tcs.2008.09.065}}.

\bibitem{Fishburn1985}
Peter~C. Fishburn.
\newblock {\em {Interval Orders and Interval Graphs: A Study of Partially Ordered Sets}}.
\newblock Wiley, 1985.

\bibitem{DBLP:journals/algorithmica/Francis0O15}
Mathew~C. Francis, Daniel Gon{\c{c}}alves, and Pascal Ochem.
\newblock The maximum clique problem in multiple interval graphs.
\newblock {\em Algorithmica}, 71(4):812--836, 2015.
\newblock \href {https://doi.org/10.1007/s00453-013-9828-6} {\path{doi:10.1007/s00453-013-9828-6}}.

\bibitem{DBLP:conf/wg/GambetteV07}
Philippe Gambette and St{\'{e}}phane Vialette.
\newblock On restrictions of balanced 2-interval graphs.
\newblock In {\em {WG} 2007}, volume 4769 of {\em LNCS}, pages 55--65. Springer, 2007.
\newblock \href {https://doi.org/10.1007/978-3-540-74839-7\_6} {\path{doi:10.1007/978-3-540-74839-7\_6}}.

\bibitem{DBLP:conf/wg/GambetteV07long}
Philippe Gambette and St{\'{e}}phane Vialette.
\newblock On restrictions of balanced 2-interval graphs.
\newblock In Andreas Brandst{\"{a}}dt, Dieter Kratsch, and Haiko M{\"{u}}ller, editors, {\em Graph-Theoretic Concepts in Computer Science, 33rd International Workshop, {WG} 2007, Dornburg, Germany, June 21-23, 2007. Revised Papers}, volume 4769 of {\em LNCS}, pages 55--65. Springer, 2007.
\newblock \href {https://doi.org/10.1007/978-3-540-74839-7\_6} {\path{doi:10.1007/978-3-540-74839-7\_6}}.

\bibitem{DBLP:journals/dm/Gardi07}
Fr{\'{e}}d{\'{e}}ric Gardi.
\newblock The {R}oberts characterization of proper and unit interval graphs.
\newblock {\em Discret. Math.}, 307(22):2906--2908, 2007.
\newblock \href {https://doi.org/10.1016/j.disc.2006.04.043} {\path{doi:10.1016/j.disc.2006.04.043}}.

\bibitem{DBLP:journals/siammax/GriggsW80}
Jerrold~R. Griggs and Douglas~B. West.
\newblock Extremal values of the interval number of a graph.
\newblock {\em {SIAM} J. Algebraic Discret. Methods}, 1(1):1--7, 1980.
\newblock \href {https://doi.org/10.1137/0601001} {\path{doi:10.1137/0601001}}.

\bibitem{DBLP:journals/tcs/HabibMPV00}
Michel Habib, Ross~M. McConnell, Christophe Paul, and Laurent Viennot.
\newblock Lex-{BFS} and partition refinement, with applications to transitive orientation, interval graph recognition and consecutive ones testing.
\newblock {\em Theor. Comput. Sci.}, 234(1-2):59--84, 2000.
\newblock \href {https://doi.org/10.1016/S0304-3975(97)00241-7} {\path{doi:10.1016/S0304-3975(97)00241-7}}.

\bibitem{DBLP:journals/tcs/Jiang10}
Minghui Jiang.
\newblock On the parameterized complexity of some optimization problems related to multiple-interval graphs.
\newblock {\em Theor. Comput. Sci.}, 411(49):4253--4262, 2010.
\newblock \href {https://doi.org/10.1016/j.tcs.2010.09.001} {\path{doi:10.1016/j.tcs.2010.09.001}}.

\bibitem{DBLP:journals/algorithmica/Jiang13}
Minghui Jiang.
\newblock Recognizing d-interval graphs and d-track interval graphs.
\newblock {\em Algorithmica}, 66(3):541--563, 2013.
\newblock \href {https://doi.org/10.1007/s00453-012-9651-5} {\path{doi:10.1007/s00453-012-9651-5}}.

\bibitem{DBLP:journals/tcs/JiangZ12}
Minghui Jiang and Yong Zhang.
\newblock Parameterized complexity in multiple-interval graphs: Domination, partition, separation, irredundancy.
\newblock {\em Theor. Comput. Sci.}, 461:27--44, 2012.
\newblock \href {https://doi.org/10.1016/j.tcs.2012.01.025} {\path{doi:10.1016/j.tcs.2012.01.025}}.

\bibitem{joseph1992determining}
Deborah Joseph, Joao Meidanis, and Prasoon Tiwari.
\newblock Determining {DNA} sequence similarity using maximum independent set algorithms for interval graphs.
\newblock In {\em Scandinavian Workshop on Algorithm Theory}, pages 326--337. Springer, 1992.

\bibitem{lekkeikerker1962representation}
C~Lekkerkerker and Johan Boland.
\newblock Representation of a finite graph by a set of intervals on the real line.
\newblock {\em Fundamenta Mathematicae}, 51(1):45--64, 1962.

\bibitem{mcguigan1977presentation}
Robert McGuigan.
\newblock {Presentation at NSF-CBMS Conference at Colby College}, 1977.

\bibitem{mckee1999topics}
Terry~A McKee and Fred~R McMorris.
\newblock {\em Topics in intersection graph theory}.
\newblock SIAM, 1999.

\bibitem{zbMATH03307330}
Fred~S. Roberts.
\newblock Indifference graphs.
\newblock In F.~Harary, editor, {\em Proof Techniques in Graph Theory}, pages 139--146. Academic Press, NY, 1969.

\bibitem{roberts1978graph}
Fred~S. Roberts.
\newblock {\em Graph theory and its applications to problems of society}.
\newblock SIAM, 1978.

\bibitem{scheinerman1983interval}
Edward~R Scheinerman and Douglas~B West.
\newblock The interval number of a planar graph: Three intervals suffice.
\newblock {\em J. Comb. Theory, Ser. {B}}, 35(3):224--239, 1983.

\bibitem{alexandre}
Alexandre Simon.
\newblock Algorithmic study of 2-interval graphs.
\newblock Master's thesis, Delft University of Technology, 2021.

\bibitem{DBLP:journals/jgt/TrotterH79}
William~T. Trotter and Frank Harary.
\newblock On double and multiple interval graphs.
\newblock {\em J. Graph Theory}, 3(3):205--211, 1979.
\newblock \href {https://doi.org/10.1002/jgt.3190030302} {\path{doi:10.1002/jgt.3190030302}}.

\bibitem{vialette2004computational}
St{\'e}phane Vialette.
\newblock On the computational complexity of 2-interval pattern matching problems.
\newblock {\em Theor. Comput. Sci.}, 312(2-3):223--249, 2004.
\newblock \href {https://doi.org/10.1016/j.tcs.2003.08.010} {\path{doi:10.1016/j.tcs.2003.08.010}}.

\bibitem{DBLP:journals/dam/WestS84}
Douglas~B. West and David~B. Shmoys.
\newblock Recognizing graphs with fixed interval number is {NP}-complete.
\newblock {\em Discret. Appl. Math.}, 8(3):295--305, 1984.
\newblock \href {https://doi.org/10.1016/0166-218X(84)90127-6} {\path{doi:10.1016/0166-218X(84)90127-6}}.

\bibitem{DBLP:journals/tcs/YamazakiSKU20}
Kazuaki Yamazaki, Toshiki Saitoh, Masashi Kiyomi, and Ryuhei Uehara.
\newblock Enumeration of nonisomorphic interval graphs and nonisomorphic permutation graphs.
\newblock {\em Theor. Comput. Sci.}, 806:310--322, 2020.
\newblock URL: \url{https://doi.org/10.1016/j.tcs.2019.04.017}, \href {https://doi.org/10.1016/J.TCS.2019.04.017} {\path{doi:10.1016/J.TCS.2019.04.017}}.

\end{thebibliography}
\end{document}